\documentclass[12pt,a4paper]{article}
\usepackage[T1]{fontenc}
\usepackage{lmodern}
\usepackage{amsmath,amssymb,amsthm}
\usepackage{geometry}
\usepackage{hyperref}
\usepackage{enumitem}
\usepackage{booktabs}
\usepackage{natbib}
\usepackage{graphicx}
\newtheorem{theorem}{Theorem}[section]
\newtheorem{lemma}[theorem]{Lemma}
\newtheorem{proposition}[theorem]{Proposition}

\newtheorem{definition}{Definition}[section]

\newtheorem{remark}{Remark}[section]
\newtheorem{assumption}{Assumption}

\title{Feasible Sets and the Transformation of Values}
\author{Jiyuan Lyu}
\date{\today}

\begin{document}
	
	\maketitle
	
	\begin{abstract}
		This paper proposes a change in perspective on the ``transformation of values'' problem: from ``searching for a single constant solution'' to ``characterizing the allocation space under objective constraints imposed by the physical production network.'' Building an input--output model, we show mathematically that whenever the macroeconomy features a physical surplus, the set of skilled-to-simple labor reduction vectors that can sustain the subsistence floor of the entire labor force forms a bounded \emph{value-feasible set}. Within this multidimensional region, the classical ``two great macro equalities'' necessarily hold simultaneously for a reasonable range of profit rates. Hence, without violating the physical minimum conditions for reproduction, the law of value and the nominal price system can be made logically consistent.
		
		\vspace{1ex}
		
	\end{abstract}
	\section{Introduction}
	
	In Volume I of \emph{Capital}, Marx notes that ``more complex labour counts only as intensified, or rather multiplied simple labour,'' and that the proportions by which qualitatively different labours are reduced to a common unit ``are determined behind the backs of the producers by a social process'' \cite{marx_capital_2018}. This reduction mechanism from skilled (complex) labor to simple labor constitutes an indispensable cornerstone of the Marxian labour theory of value (LTV). Yet, throughout the development of formal and mathematical political economy, the reduction of heterogeneous labour has long been treated as a theoretical bottleneck for LTV. Without a rigorous and consistent reduction scale, aggregate social value cannot be meaningfully summed as a scalar, and the social connection between production and exchange uncovered by Marx risks becoming analytically disconnected \cite{weber_necessity_2019}.
	
	For a long time, the theoretical deadlock of skilled-labor reduction has stemmed from its persistent logical circularity. With the Sraffa--Leontief ``production of commodities by means of commodities'' system being rigorously formalized \cite{sraffa_production_2016, kurz_theory_1997, pasinetti_lectures_1977}, neo-Ricardian scholars such as Samuelson \cite{samuelson_understanding_1971} and Steedman \cite{steedman_heterogeneous_2019} launched sharp logical critiques against LTV. The core claim is: if one determines reduction coefficients by nominal wage differentials observed in labor markets, then the category of ``value'' becomes a redundant theoretical detour. This is because material technical conditions and the real wage are already sufficient to determine prices of production and the profit rate, leaving value conceptually suspended.
	
	To break this endogenous ``value--price--wage'' loop, Analytical Marxists and scholars in the classical political economy camp have undertaken substantial theoretical efforts. They have sought alternative objective indicators independent of the price system, including reductions based on education and training costs \cite{rowthorn_skilled_1974}, endogenous solutions from simultaneous equation systems \cite{bowles_marxian_1977}, isolating actual labor contents in joint production models \cite{flaschel_actual_1983}, and advanced tools from matrix algebra and eigenvalue theory \cite{krause_money_2020}. Although these approaches have become increasingly sophisticated mathematically, they share a hidden a priori limitation: they attempt to find a \textbf{unique and deterministic} reduction-coefficient vector. Yet any attempt to mechanically derive a unique benchmark solely from objective technical conditions ultimately runs into algebraic contradictions or violates nonnegativity constraints. As a result, the two great macro equalities typically cannot both hold when transforming values into prices of production, thereby generating the long-standing classical ``transformation problem'' \cite{marx_capital_1992}.
	
	This paper departs from the traditional algebraic quest for a ``unique benchmark'' and offers a paradigm shift toward a geometric perspective. Marx emphasized that reduction proportions are determined by a ``social process,'' which suggests that reduction coefficients should not be modeled mathematically as fixed constants; rather, they should be characterized as \textbf{degrees of freedom} jointly shaped by class struggle, historical conventions, and institutional arrangements. Accordingly, the central theoretical question should become: under the rigid constraints of objective technology and subsistence floors, what are the \textbf{objective boundaries and the existence space (the ``feasible set,'' in our terminology)} of reduction coefficients compatible with capitalist reproduction?
	
	Methodologically, this study follows the tradition of classical mathematical political economy. Since Roemer \cite{roemer1981analytical} laid a rigorous analytical foundation for Marxian economics using convex analysis, and Flaschel \cite{flaschel2010topics} systematically reconstructed classical macro--micro models with matrix algebra, nonnegative matrices and convex-cone theory have become core tools for analyzing the law of value. Inspired by Morishima's \cite{morishima_marxs_1973} classic theory and recent advances in general convex-cone economics \cite{yoshihara_class_2010, veneziani_analytical_2012}, this paper develops a framework based on nonnegative matrices and convex geometry, integrating the physical production network with heterogeneous reproduction requirements of labor power. The main contributions are as follows.
	
	First, we prove rigorously that as long as the economy has a positive surplus, the set of reduction coefficients that can support social reproduction constitutes a bounded closed convex set, i.e., the \textbf{value-feasible set}. This provides a benchmark for skilled-labor reduction that is immune to distortions arising from profit rates in the sphere of circulation.
	
	Second, we propose a \textbf{two-layer homomorphic mapping} mechanism that yields a new approach to the transformation problem. The traditional difficulty arises from trying to force ``the generation of values'' and ``the nominal distribution via prices'' into a single mathematical space. Using the algebraic structure of the input--output matrix, we establish the existence of a profit-rate compatibility interval such that the transformation hyperplane necessarily intersects the value-feasible set. Once the geometry of intersection is determined, the value system is projected into a price space that incorporates profit-rate compression by introducing a monetary-dimension multiplier. This approach differs fundamentally---while respecting the reproduction floor---from the ``New Interpretation'' (NI) that relaxes the physical constraint that ``workers must consume a specific bundle of physical goods'' \cite{foley_value_1982, dumenil_beyond_1983, dumenil_labor_1989}, and it is also distinct from the ``Temporal Single-System Interpretation'' (TSSI) that attempts reconciliation by introducing historical time \cite{kliman_temporal_1999, veneziani_dynamics_2005}. Our framework confirms the consistency of Marx's macro equalities within a static production network.
	
	The remainder of the paper is organized as follows. Sections 2 and 3 define the heterogeneous consumption framework and construct the basic labor-reproduction matrix. Sections 4 and 5 rigorously define the value-feasible set and present its properties and existence theorem. Section 6 introduces the profit rate and reveals the compression effect on the space of price distributions. Section 7 reformulates and solves the transformation problem using the two-layer mapping. Section 8 provides a numerical example. The final section concludes.
	\section{Theoretical Model and Basic Assumptions}
	
	\subsection{Economic structure}
	
	Consider a closed economy with $n \geq 2$ sectors. To describe input--output relations clearly, assume that sector $j$ produces only commodity $j$ and uses exclusively the concrete labor of type $j$. Joint production is not considered for now.
	
	In modern large-scale production, the production process not only consumes current intermediate inputs but also requires the advance of substantial fixed capital such as machinery and equipment. Marx's labour theory of value holds that undepreciated machinery does not transfer value to the new product; only the current depreciation transfers as ``dead labour'' into the value of output. Hence, the model must explicitly distinguish between \emph{flow} usage and \emph{stock} advances.
	
	\begin{definition}[Intermediate input and capital stock matrices]
		The physical intermediate input matrix $A=(a_{ij})$ is an $n\times n$ nonnegative matrix, where $a_{ij}$ denotes the quantity of commodity $i$ (a flow) required as intermediate input to produce one unit of commodity $j$.
		The capital stock matrix $K=(\kappa_{ij})$ is an $n\times n$ nonnegative matrix, where $\kappa_{ij}$ denotes the physical stock of commodity $i$ that must be advanced in order to produce one unit of commodity $j$.
	\end{definition}
	
	\begin{definition}[Fixed-capital depreciation matrix]
		Let $\delta_j \in (0, 1]$ be the depreciation rate of fixed capital in sector $j$, and define the diagonal depreciation matrix $\hat{\delta} = \mathrm{diag}(\delta_1,\ldots,\delta_n)$. The current-period depreciation (i.e., transfer of dead labour) matrix is $D = K\hat{\delta}$, whose element $d_{ij} = \kappa_{ij}\delta_j$ represents the wear and tear allocated to producing one unit of commodity $j$.
	\end{definition}
	
	To measure total real material usage, define the composite material input matrix $\tilde{A} = A + D$.
	
	\begin{definition}[Labor and consumption matrices]
		The labor-coefficient matrix $L=\mathrm{diag}(l_1,\ldots,l_n)$ is a diagonal matrix, where $l_j>0$ denotes the direct labor time (hours) of sector $j$ required to produce one unit of commodity $j$.
		The consumption-coefficient matrix $B=(\beta_{ij})$ is an $n\times n$ matrix, where $\beta_{ij}>0$ denotes the amount of commodity $i$ that a worker in sector $j$ must consume per hour of work to reproduce her labor power.
	\end{definition}
	
	\subsection{Technological and consumption constraints}
	
	To ensure the model has meaningful economic content, we impose the necessary technological and subsistence restrictions on the above matrices.
	
	\begin{assumption}[Technology and labor characteristics]\label{ass:A_L}
		The composite material input matrix $\tilde{A}$ is nonnegative ($\tilde{A}\geq 0$), irreducible (i.e., there exist direct or indirect production linkages across sectors), and productive in the sense that the dominant eigenvalue satisfies $\lambda_{\max}(\tilde{A})<1$ (the Hawkins--Simon condition). Moreover, each sector requires direct labor: for all $j$, $l_j>0$.
	\end{assumption}
	
	\begin{assumption}[Consumption-network properties]\label{ass:B}
		The consumption-coefficient matrix $B \geq 0$ and has no all-zero column. That is, for each $j = 1,\ldots,n$, there exists some $i$ such that $\beta_{ij} > 0$.
	\end{assumption}
	
	It should be emphasized that Assumption~\ref{ass:B} allows $B$ to be highly sparse. To reproduce labor power, workers in each sector must consume certain necessities (e.g., food and clothing), but they need not directly consume all industrial products (e.g., machine tools or pig iron). Because the production network contains indirect linkages, even if workers consume only a small set of final goods, the induced economy-wide input requirements along the full supply chain suffice to connect aggregate labor networks to the reproduction of labor power in each sector.
	
	\subsection{The boundary of simple reproduction}
	
	Under capitalist simple reproduction (zero profits), prices of production exactly cover material input costs (including intermediates and depreciation) and wage costs. Let the price vector be $\mathbf{p}=(p_1,\ldots,p_n)^T>0$ and the wage vector $\mathbf{w}=(w_1,\ldots,w_n)^T>0$ (where $w_j$ is the nominal hourly wage in sector $j$). The corresponding price system satisfies:
	\begin{equation}\label{eq:price}
		\mathbf{p}^T = \mathbf{p}^T A + \mathbf{p}^T D + \mathbf{w}^T L = \mathbf{p}^T \tilde{A} + \mathbf{w}^T L
	\end{equation}
	
	On the other hand, the material and social reproduction of labor power constitutes a hard floor of the economic circuit. Since labor complexity differs across sectors, the commodity basket required to reproduce labor power (i.e., the columns of $B$) is also heterogeneous. For the system to persist, the nominal hourly wage $w_j$ must be sufficient to purchase the minimum subsistence basket for that particular type of labor power, with expenditure $\sum_{i=1}^n p_i \beta_{ij}$. Aggregating the subsistence floors across all $n$ sectors yields the vector form of the minimum-wage constraint:
	\begin{equation}\label{eq:wage_constraint}
		\mathbf{w}^T \geq \mathbf{p}^T B
	\end{equation}
	
	\section{The Labor-Reproduction Matrix}
	
	\subsection{Construction}
	
	From the price equation \eqref{eq:price} and the productivity assumption on the composite material input matrix, $(I-\tilde{A})$ is invertible. We can thus solve for the commodity-price vector priced in terms of the wage vector:
	\begin{equation}\label{eq:price_solution}
		\mathbf{p}^T = \mathbf{w}^T L(I-\tilde{A})^{-1}
	\end{equation}
	
	Substituting this price solution into the minimum-wage constraint \eqref{eq:wage_constraint} yields $\mathbf{w}^T \geq \mathbf{w}^T L(I-\tilde{A})^{-1} B$. For subsequent analysis, we define the key operator endogenous to technology and consumption structure.
	
	\begin{definition}[Labor-reproduction matrix]\label{def:M}
		Define the basic labor-reproduction matrix as $M_0 = L(I-\tilde{A})^{-1}B$. The wage constraint sustaining the economic circuit can be compactly written as:
		\begin{equation}\label{eq:wage_M}
			\mathbf{w}^T(I-M_0) \geq \mathbf{0}^T
		\end{equation}
	\end{definition}
	
	\subsection{Dimensionlessness and price invariance}
	\begin{proposition}[Dimensionlessness of $M_0$]\label{prop:dimensionless}
		All elements of $M_0$ are dimensionless; their magnitudes purely reflect nested proportions of labor time.
	\end{proposition}
	
	Expanding $[M_0]_{ij} = l_i \sum_{k} [(I-\tilde{A})^{-1}]_{ik} \beta_{kj}$, we see that the direct labor coefficient $l_i$ has units of [hours/unit of output], the Leontief inverse captures interrelations in [units of input per unit of output], and the consumption coefficient $\beta_{kj}$ corresponds to [units of consumption per hour]. Multiplying them cancels all physical units, leaving a pure number that represents: ``to reproduce one hour of labor power in sector $j$, how many hours of sector $i$'s labor must be expended, directly or indirectly, throughout the economy-wide supply-chain network.''
	
	As an objective magnitude, the labor-reproduction matrix is independent of any humanly chosen unit of account.
	
	\begin{lemma}[Price invariance of $M_0$]\label{lem:price_invariance}
		Let $\mathbf{p}>0$ be any strictly positive price vector, and let $P=\mathrm{diag}(p_1,\ldots,p_n)$. If we convert the original physical matrices into monetary matrices measured at this price vector (i.e., $\tilde{A}_{val}=P\tilde{A}P^{-1}$, $\tilde{L}=LP^{-1}$, $\tilde{B}=PB$), then the constructed matrix $M_0$ is exactly identical to the one computed in pure physical units.
	\end{lemma}
	
	In real-world macroeconomic statistics, it is difficult to obtain input--output tables in pure physical units; researchers often have access only to tables measured in current monetary prices. The price-invariance lemma ensures methodologically that whether commodities are measured in tons and liters or in currency units, the extracted labor-reproduction matrix $M_0$ and its algebraic properties remain invariant.
	
	\subsection{Network connectivity and spectral properties}
	
	Given the above algebraic properties, $M_0$ exhibits strong network connectivity. Since $\tilde{A}$ with fixed-capital depreciation remains nonnegative and convergent like the standard Leontief matrix, the structural features of the labor-reproduction operator are fully preserved.
	
	\begin{lemma}[Basic spectral properties of $M_0$]\label{lem:M_spectrum}
		Under Assumptions~\ref{ass:A_L} and~\ref{ass:B}, the labor-reproduction matrix satisfies $M_0 > 0$ (a strictly positive matrix) and is irreducible and aperiodic. Hence its dominant Perron--Frobenius eigenvalue $\lambda^* = \lambda_{\max}(M_0)$ is a simple positive real number, associated with unique strictly positive left and right eigenvectors. Moreover, there exists a strict spectral gap.
	\end{lemma}
	
	The strict positivity of $M_0$ (i.e., no element is zero) captures the indivisibility of the modern division of labor. As discussed above, even if coal miners do not directly consume textiles, the grain they buy requires fertilizer; fertilizer production requires chemical equipment; and manufacturing that equipment uses textile-sector products. Such indirect production chains are infinitely linked through $(I-\tilde{A})^{-1}$. Objectively, reproducing labor power in any sector indirectly absorbs the labor of all other sectors in the economy. This all-round interdependence establishes the algebraic status of $M_0$ as a strictly positive matrix and paves the way for identifying a benchmark solution (the Perron eigenvector) in what follows.
	\section{Skilled-Labor Reduction}
	
	After constructing the basic labor-reproduction matrix $M_0$---which incorporates intermediate-input consumption and fixed-capital depreciation---we can formally translate Marx's hypothesis of skilled-labor (complex-labor) reduction into a rigorous algebraic constraint. The analysis below shows that, given objective technology and subsistence-floor conditions, society cannot arbitrarily choose reduction coefficients; instead, they are tightly confined to a specific region.
	
	\subsection{Reduction coefficients}
	
	To measure heterogeneous labor in a common unit, let $\mathbf{c}=(c_1,c_2,\ldots,c_n)^T>0$ be an exogenous vector of reduction coefficients. Without loss of generality, we take sector 1 labor as the benchmark simple labor and set $c_1=1$. This means that when a worker in sector $j$ works for one hour, she creates new value equivalent to $c_j$ hours of simple labor.
	
	Under the labour theory of value, the socially necessary labor time embodied in commodities (i.e., labor values $\mathbf{v}$) consists of two components: (i) old value transferred from means of production, and (ii) new value created by living labor. With fixed capital, the transfer of old value includes not only current intermediate inputs ($A$) but also current depreciation of equipment ($D$). Hence the value system of the whole economy is determined by the following system:
	\begin{equation}
		\mathbf{v}^T = \mathbf{v}^T A + \mathbf{v}^T D + \mathbf{c}^T L = \mathbf{v}^T \tilde{A} + \mathbf{c}^T L
	\end{equation}
	Solving yields a unique value system:
	\begin{equation}
		\mathbf{v}^T = \mathbf{c}^T L(I-\tilde{A})^{-1}
	\end{equation}
	
	We next consider a central concept in Marxian economics: the origin of surplus value. The value of labor power per hour for workers in sector $j$ equals the total value of the commodity bundle that must be consumed to reproduce that labor power. Using the previously defined matrix $M_0$, this can be written compactly as:
	\begin{equation}
		\sigma_j = \mathbf{v}^T \mathbf{b}_j = \mathbf{c}^T L(I-\tilde{A})^{-1}\mathbf{b}_j = [\mathbf{c}^T M_0]_j
	\end{equation}
	where $\mathbf{b}_j$ is the $j$-th column of the heterogeneous consumption matrix $B$.
	
	According to exploitation theory, surplus value arises from the gap between new value created ($c_j$) and the value of labor power ($\sigma_j$). We therefore define the objective exploitation rate of sector $j$ in the pure value space as:
	\begin{equation}\label{eq:exploitation_rate}
		e_j(\mathbf{c}) = \frac{c_j - [\mathbf{c}^T M_0]_j}{[\mathbf{c}^T M_0]_j}
	\end{equation}
	
	To prevent the macroeconomic network from breaking down, the most basic physical reproduction floor requires that in every sector the new value created by workers must at least compensate the value of their own labor-power consumption. That is, for all $j$, the nonnegative-exploitation condition must hold: $c_j \geq [\mathbf{c}^T M_0]_j$. In matrix form, this condition is $\mathbf{c}^T(I-M_0)\geq\mathbf{0}^T$.
	
	\subsection{Constructing the value-feasible domain}
	
	Based on the macro-level physical reproduction constraints above, we now formally delimit the region in which reduction coefficients may reside.
	
	\begin{definition}[Value Feasible Domain]\label{def:Theta_val}
		Given composite technological conditions and the subsistence consumption structure $(\tilde{A},L,B)$ (including fixed-capital consumption), the set of reduction coefficients that satisfy the physical reproduction floor for the entire labor force is defined as:
		\begin{equation}
			\Theta^{val} = \left\{\mathbf{c}\in\mathbb{R}^n_{++} :
			c_1=1,\; \mathbf{c}^T(I-M_0)\geq\mathbf{0}^T\right\}
		\end{equation}
	\end{definition}
	
	$\Theta^{val}$ constitutes the rigid boundary within which a social system can choose when converting complex labor into simple labor. Outside this boundary, there exists at least one sector in which workers create value insufficient even to cover the value of the consumption goods necessary for their own subsistence. In that case, labor power in that sector cannot be reproduced, and the macro supply-chain network collapses.
	
	\begin{remark}[Strict feasibility in the interior]\label{rem:strict_interior}
		In a real capitalist economy, capitalists in all sectors necessarily demand strictly positive surplus value (i.e., $c_j > [\mathbf{c}^T M_0]_j$ for all $j$). In that case, the reduction coefficients must lie in the interior of the value-feasible domain (the strictly feasible region $\Theta^\circ = \{\mathbf{c} : c_1=1, \mathbf{c}^T(I-M_0)>\mathbf{0}^T\}$). However, in our subsequent mathematical analysis, in order to work with a closed set (supporting compactness-like arguments for extrema and intersections), we consistently conduct the discussion on the closed set with boundary, $\Theta^{val}$.
	\end{remark}
	
	The feasible-domain perspective implies that, as long as the objective physical network generates a positive surplus, the conversion vector $\mathbf{c}$ determined by the ``social process'' possesses substantial degrees of freedom.
	\section{Properties of the Value-Feasible Domain and a Benchmark Solution}
	
	After defining the value-feasible domain, a natural question arises: under what conditions is this set nonempty? If it exists, what geometric properties does it have? This section shows that the existence of the domain is closely tied to algebraic features of the objective physical network, and that its interior contains a benchmark with special economic meaning.
	
	\subsection{Spectral equivalence for the existence theorem}
	
	The ``Fundamental Marxian Theorem'' (FMT) establishes that a positive exploitation rate is necessary and sufficient for the existence of positive profits. In our heterogeneous-consumption framework with fixed-capital consumption, this celebrated theorem can be extended and sharpened into an eigenvalue-equivalence proposition.
	
	\begin{theorem}[Existence theorem for reduction coefficients]\label{thm:existence}
		Under Assumptions~\ref{ass:A_L} and~\ref{ass:B}, the following three conditions concerning the economy's capacity to generate surplus value are equivalent:
		\begin{enumerate}[label=(\roman*)]
			\item The value-feasible domain has a nonempty interior (i.e., there exists at least one vector of reduction coefficients that yields strictly positive exploitation in every sector, $\operatorname{int}(\Theta^{val})\neq\emptyset$);
			\item The dominant eigenvalue of the labor-reproduction matrix is strictly less than one ($\lambda_{\max}(M_0)< 1$);
			\item The extended physical matrix that includes the class-based material advance satisfies productivity, i.e., the economy features a macro-level physical surplus ($\lambda_{\max}(\tilde{A}+BL)< 1$).
		\end{enumerate}
		\textit{(On the zero-surplus boundary, replacing the strict inequalities above by ``$\leq 1$'' corresponds to the nonemptiness of the closed set $\Theta^{val}\neq\emptyset$. For a matrix-theoretic proof, see the Appendix.)}
	\end{theorem}
	
	The existence theorem links the invisible value-exploitation relation (Condition i) to directly measurable features of the physical production network (Condition iii) through matrix algebra (Condition ii). It implies that if a society can still produce a physical surplus after replacing intermediate inputs, equipment depreciation, and workers' subsistence goods, then regardless of how complex social conflicts generate large cross-sector differences in reduction coefficients, there must exist an entire region of reduction-coefficient vectors that allow capitalist value valorization.
	
	\subsection{Boundedness, closedness, and convexity of the feasible region}
	
	Once nonemptiness is established, we must further examine how far the solution set extends in multidimensional space. Does ``free choice'' of reduction coefficients mean they can grow without bound?
	
	\begin{proposition}[Bounded closed convexity of $\Theta^{val}$]\label{prop:convex}
		When the economy has a surplus ($\lambda_{\max}(M_0) < 1$), and sector 1 is chosen as the benchmark ($c_1=1$), the value-feasible domain $\Theta^{val}$ forms, in the remaining $n-1$ dimensions, a \textbf{bounded closed convex set with a strictly positive lower bound}.
		\textit{(A strict algebraic proof deriving upper and lower bounds in each dimension via diagonal elements is provided in the Appendix.)}
	\end{proposition}
	
	Because the feasible domain has a well-defined upper bound in every dimension, no type of labor (e.g., complex labor in high-tech sectors) can demand an infinitely large value-conversion coefficient under a given technological network. At the same time, the strictly positive lower bound protects simple labor: even if social institutions strongly favor mental labor, the reduction coefficient of other labor types cannot be driven down toward zero. The entire social evaluation system of labor is thus strictly confined within boundaries determined by physical requirements (matrices $\tilde{A}$ and $B$). For later exposition, we refer to choosing sector 1 as the benchmark as normalization.
	
	\subsection{Eigenvectors}
	
	Having established that the feasible domain is a bounded full-dimensional polytope, we seek a meaningful reference point within it. Since the labor-reproduction matrix $M_0$ is strictly positive, by the Perron--Frobenius theorem its dominant eigenvalue $\lambda_{\max}(M_0)$ corresponds to a strictly positive left eigenvector $\mathbf{y}^*$.
	
	\begin{theorem}[Equal-exploitation point]\label{thm:equilibrium}
		Let $\lambda^*=\lambda_{\max}(M_0)<1$, and let $\mathbf{y}^*$ be the normalized ($y_1^*=1$) positive eigenvector of $M_0$. Then $\mathbf{y}^*$ lies in the interior of the value-feasible domain ($\mathbf{y}^* \in \operatorname{int}(\Theta^{val})$). Moreover, $\mathbf{y}^*$ is the \textbf{unique} reduction-coefficient vector in the solution set that equalizes objective exploitation rates across all sectors. In that case, there exists a common economy-wide exploitation rate $e^*=(1-\lambda^*)/\lambda^*>0$.
		\textit{(For the eigenvalue mapping and the algebraic proof of uniqueness, see the Appendix.)}
	\end{theorem}
	
	This ``equal-exploitation point'' $\mathbf{y}^*$ need not be the terminal state toward which capitalist markets spontaneously converge. In reality, due to differences in the organic composition of capital, bargaining power, and the presence of monopolistic forces, actual exploitation rates across sectors must be uneven.
	
	Nevertheless, the existence of $\mathbf{y}^*$ provides an analytical benchmark. On the one hand, it is entirely endogenous to the objective physical production network and excludes institutional frictions; on the other hand, it represents an ideal state in which surplus extraction is most ``uniform'' across sectors.
	
	\subsection{The zero-surplus boundary}
	
	When technology is inefficient or subsistence needs are sufficiently high so that the macroeconomy can exactly cover all costs but cannot accumulate any surplus (i.e., $\lambda_{\max}(M_0)=1$), the feasible domain defined above shrinks dramatically. We characterize this case as follows.
	
	\begin{proposition}[Degeneration of the feasible domain]\label{prop:collapse}
		When $\lambda_{\max}(M_0)=1$:
		\begin{enumerate}[label=(\roman*)]
			\item The interior of the value-feasible domain becomes empty; that is, no strictly positive exploitation is possible;
			\item The closed set $\Theta^{val}$ degenerates into a unique singleton set $\Theta^{val} = \{\mathbf{y}^*\}$, where $\mathbf{y}^*$ is the unique normalized ($y_1^*=1$) positive eigenvector of $M_0$;
			\item Under this unique reduction vector, the realized exploitation rate in every sector is zero ($e_j = 0, \forall j$).
		\end{enumerate}
		\textit{(An algebraic proof using eigenvector orthogonality is provided in the Appendix.)}
	\end{proposition}
	
	At the break-even line of material scarcity, objective technological constraints completely lock in what would otherwise be the bargaining space afforded to the ``social process.'' Society then loses any freedom to choose reduction coefficients: the only admissible distributive scheme is uniquely determined by the underlying physical input requirements.
	
	\section{How the Profit Rate Constrains the Distributive Space}
	
	The previous analysis was built on objective technology under the assumption of zero profits, thereby constructing the foundation for value formation---the value-feasible domain $\Theta^{val}$. In an actual capitalist economy, however, production is not only a process of creating material goods but also a process of realizing and distributing value. Capitalists in each sector claim a uniform profit rate $r$ on the basis of the total fixed-capital stock they advance ($K$). This section shows how capital's appropriation of a profit rate compresses, at the algebraic level, the economy-wide set of feasible nominal wage structures.
	
	\subsection{A price system with stock-based profits}
	
	When there is an economy-wide uniform profit rate $r \geq 0$, the equation for prices of production becomes the sum of costs and stock-based profits:
	\begin{equation}\label{eq:price_profit}
		\mathbf{p}^T = \underbrace{\mathbf{p}^T \tilde{A}}_{\text{composite material costs}} + \underbrace{\mathbf{w}^T L}_{\text{wage costs}} + \underbrace{r \mathbf{p}^T K}_{\text{profits on advanced stocks}}
	\end{equation}
	
	Given physical technology, there exists a purely technical upper bound on the profit rate, $r_A$ (defined such that the dominant eigenvalue of $\tilde{A} + r_A K$ equals 1). For any $r \in [0, r_A)$, the above system admits a unique strictly positive price solution:
	\begin{equation}
		\mathbf{p}^T = \mathbf{w}^T L[I - \tilde{A} - r K]^{-1}
	\end{equation}
	
	Substituting this price solution into the workers' subsistence constraint $\mathbf{w}^T \geq \mathbf{p}^T B$ yields a generalized labor-reproduction operator parameterized by the profit rate.
	
	\begin{definition}[Parametric labor-reproduction matrix]\label{def:M_r}
		For any $r \in [0, r_A)$, define the parametric labor-reproduction matrix:
		\begin{equation}\label{eq:M_r}
			M(r) := L[I - \tilde{A} - r K]^{-1}B
		\end{equation}
		Then the nominal-wage constraint sustaining the reproduction of labor power can be written as $\mathbf{w}^T[I - M(r)] \geq \mathbf{0}^T$.
	\end{definition}
	
	After introducing the stock-based profit-rate parameter, this generalized matrix has well-behaved algebraic properties, summarized in the proposition below.
	
	\begin{proposition}[Monotonic properties of the parametric reproduction matrix]\label{prop:M_r_properties}
		Under Assumptions~\ref{ass:A_L} and~\ref{ass:B}, for any $r \in [0, r_A)$:
		\begin{enumerate}[label=(\roman*)]
			\item $M(r) > 0$ (a strictly positive matrix) and is continuous in $r$;
			\item as $r$ increases, every element of $M(r)$ increases strictly;
			\item the dominant eigenvalue $\lambda_{\max}(M(r))$ increases strictly in $r$, and $\lambda_{\max}(M(r)) \to +\infty$ as $r \to r_A$.
		\end{enumerate}
		\textit{(A proof based on the series expansion of the inverse matrix is provided in the Appendix.)}
	\end{proposition}
	
	The ``elementwise strict monotonicity'' in the proposition encodes a class-distribution intuition: capitalists' pursuit of a higher profit rate appears as an amplification via the inverse matrix, raising production prices relative to nominal wages. As a result, to purchase the same subsistence bundle ($B$), workers require the economy-wide circulation sphere to ``advance'' a longer accounting time.
	
	\subsection{The maximum feasible profit rate}
	
	As $r$ rises, the dominant eigenvalue $\lambda_{\max}(M(r))$ increases strictly. Once this eigenvalue exceeds 1, labor power can no longer be reproduced under the corresponding price system. This motivates an upper bound on profit rates compatible with reproduction.
	
	\begin{theorem}[Existence of the maximum feasible profit rate]\label{thm:r_star}
		Assuming the system features a physical surplus ($\lambda_{\max}(M(0)) < 1$), there exists a unique critical profit rate $r^* \in (0, r_A)$ such that $\lambda_{\max}(M(r^*)) = 1$. The value $r^*$ is the upper bound on the profit rate compatible with sustained reproduction of labor power.
		\textit{(A proof based on the intermediate value theorem for continuous functions is provided in the Appendix.)}
	\end{theorem}
	
	The bound $r^*$ is strictly smaller than the purely technical maximum profit rate $r_A$. Sraffa's $r_A$ corresponds to the extreme abstraction in which wages are driven down to zero. In our framework with matrix $B$, the critical profit rate $r^*$ must leave workers with a minimum subsistence bundle. The gap $r_A - r^*$ measures the rigid constraint imposed by the reproduction floor on capital's drive for valorization.
	
	To clearly characterize the distributive range in the wage dimension, we define the feasible set of relative nominal wage structures (normalized by $w_{1}^{rel}=1$) implied by $M(r)$.
	
	\begin{definition}[Price--Wage Feasible Domain]\label{def:Theta_price}
		For a given profit rate $r \in [0, r^*]$, define the set of relative nominal wage structures that sustain reproduction as:
		\begin{equation}
			\Theta^{price}(r) = \left\{\mathbf{w}^{rel}\in\mathbb{R}^n_{++} :
			w_1^{rel}=1,\; (\mathbf{w}^{rel})^T(I-M(r))\geq\mathbf{0}^T\right\}
		\end{equation}
	\end{definition}
	
	\subsection{Profit--wage duality}
	
	By examining how the boundary of the parameterized set $\Theta^{price}(r)$ shifts, we uncover a ``profit--wage dual structure.''
	
	\begin{theorem}[Duality between the profit rate and distributive sets]\label{thm:duality}
		An increase in the capitalist profit rate has a strict crowding-out effect on the feasible set of nominal wage structures. Specifically:
		\begin{enumerate}[label=(\roman*)]
			\item \textbf{Strict monotone contraction:} for any $0 \leq r_1 < r_2 \leq r^*$, there is a strict inclusion $\Theta^{price}(r_2) \subsetneq \Theta^{price}(r_1) \subseteq \Theta^{val}$.
			\item \textbf{Singleton degeneration:} when the profit rate reaches its upper bound $r = r^*$, the strict interior disappears and the closed set $\Theta^{price}(r^*)$ degenerates to a unique point, namely the positive eigenvector of $M(r^*)$.
		\end{enumerate}
		\textit{(For rigorous proofs of these set-inclusion relations, see the Appendix.)}
	\end{theorem}
	
	This theorem describes class conflict at the macro level. A higher profit rate $r$ means a larger share of national income is appropriated by capital in accounting terms; in the algebra of sets, it means the feasible wage set $\Theta^{price}(r)$ continuously contracts inward. This shows that the expansion of profits not only reduces workers' physical share but also, more deeply, erodes society's flexibility in differentiating nominal wages across industries. As the profit rate approaches $r^*$, society loses the bargaining room to adjust distribution altogether.
	\section{Revisiting the Transformation Problem}
	
	Having clarified the two-layer algebraic structure constituted by the value-generation set $\Theta^{val}$ and the price-distribution set $\Theta^{price}(r)$, this section engages in a dialogue with the long-standing ``Transformation Problem'' in Marxian political economy and provides a rigorous compatibility result from the viewpoint of cross-space mappings.
	
	\subsection{The traditional impasse and a change in perspective}
	
	Since Bortkiewicz's (1907) classic critique, it has been widely held that, given technology, when transforming labor values into prices of production with a uniform profit rate, the two identities Marx hoped to preserve---``total price = total value'' and ``total profit = total surplus value'' (hereafter, the ``two macro equalities'')---generally cannot both hold. For decades, debates on the transformation problem have largely focused on redefining the measurement basis of values or prices so as to rescue these identities.
	
	This paper argues that the algebraic impasse in much of the traditional transformation literature stems, to a large extent, from two unnoticed a priori assumptions. First, it presumes that the skilled-to-simple labor reduction vector $\mathbf{c}$ is a single, exact number rather than a multidimensional region with substantial interior degrees of freedom. Second, it attempts to forcibly stitch together ``the physical generation of values'' and ``the nominal distribution through prices'' within one and the same mathematical space.
	
	Departing from the search for a unique constant solution, we reformulate the transformation problem as a macroeconomic mapping process across distinct spaces. In Marx's theory, the substantive creation of value takes place in the \textbf{sphere of production}, constrained by the objective physical reproduction floor (i.e., $\Theta^{val}$ as defined in this paper); by contrast, the realization of prices of production occurs in the \textbf{sphere of circulation and distribution}, constrained by capital's pursuit of an average profit rate (i.e., $\Theta^{price}(r)$).
	
	On this basis, we argue that the essence of transformation compatibility is to accomplish two layers of macro balancing in sequence. First, by introducing a ``labor--money'' conversion scalar, one ensures dimensionally that ``total price equals total value,'' aligning labor accounting with monetary accounting at the level of aggregate scale. Second, and more crucially, one must \textbf{search within the broad evaluative space permitted by the underlying physical network ($\Theta^{val}$) for those reduction structures $\mathbf{c}^*$ that make ``total surplus value'' exactly equal to ``total profit.''}
	
	Because the true reduction coefficients form a multidimensional region with interior freedom, the ensuing derivations show that as long as the macro profit rate lies in a reasonable compatibility range, reduction structures satisfying the two macro equalities not only exist objectively but are also non-unique (algebraically, they appear as a multidimensional intersection inside the feasible domain). This implies that, once skilled-labor reduction is treated as an elastic social evaluation structure, the underlying material production network naturally contains enough ``room'' for the law of value and the system of prices of production to be macroscopically consistent.
	
	\subsection{The first equality}
	
	To incorporate the macro equalities into our framework, we must aggregate social values and prices at the macro level. Given the vector of gross physical outputs $\mathbf{x} > \mathbf{0}$ and any chosen reduction vector $\mathbf{c} \in \Theta^{val}$, total labor value is:
	\begin{equation}\label{eq:value_functional}
		F(\mathbf{c}) := \mathbf{v}^T \mathbf{x} = \mathbf{c}^T L(I-\tilde{A})^{-1}\mathbf{x}
	\end{equation}
	
	On the price side, for a uniform profit rate $r \in [0, r^*]$, suppose the nominal distributive system selects a relative wage structure $\mathbf{w}^{rel} \in \Theta^{price}(r)$ that satisfies reproduction. The corresponding baseline total price of production in relative terms is
	$P^{rel}(r) = \mathbf{w}^{rel,T} L[I-\tilde{A}-rK]^{-1}\mathbf{x}$.
	
	To make the first equality ``total price = total value'' hold, introduce a scalar multiplier $\kappa > 0$ converting labor-time units into monetary units (i.e., the monetary expression of labor time, MELT), and define actual nominal wages as $\mathbf{w} = \kappa \mathbf{w}^{rel}$. Imposing $\kappa P^{rel}(r) = F(\mathbf{c})$ uniquely pins down this mapping multiplier:
	\begin{equation}
		\kappa = \frac{\mathbf{c}^T L(I-\tilde{A})^{-1}\mathbf{x}}{\mathbf{w}^{rel,T} L [I-\tilde{A}-rK]^{-1}\mathbf{x}}
	\end{equation}
	
	Since $\mathbf{w}^{rel}$ is given a priori in $\Theta^{price}(r)$, multiplying it by any positive scalar $\kappa$ does not alter the algebraic property required for workers' physical reproduction. Hence, via the scaling mapping induced by $\kappa$, the first equality is guaranteed mathematically without violating the subsistence-floor constraint.
	
	\subsection{The second equality}
	
	Given the first equality, Marx's transformation theory further requires ``total profit = total surplus value'' (the second equality).
	
	First, accounting for total surplus value must be grounded strictly in the underlying physical reproduction network (i.e., $M_0$): it equals total new value created by workers minus the objective value of labor power (the subsistence bundle):
	\begin{equation}\label{eq:total_surplus}
		S(\mathbf{c}) = \mathbf{c}^T L \mathbf{x} - \mathbf{c}^T M_0 L \mathbf{x} = \mathbf{c}^T L(I-\tilde{A})^{-1} \left[ (I-\tilde{A}-BL)\mathbf{x} \right]
	\end{equation}
	
	Second, consider total profit in the price space. Since profits are claimed on the basis of the capital stock $K$, total profit can be written as $\Pi(r) = r \mathbf{p}^T K \mathbf{x}$. To facilitate cross-space comparison, define a crucial macro scalar $\gamma(r)$: under the given profit rate and nominal distributive structure, it measures the \textbf{macro profit share} of total profit in total price of production (or total value):
	\begin{equation}
		\gamma(r) = \frac{r \mathbf{w}^{rel,T} L [I - \tilde{A} - r K]^{-1} K \mathbf{x}}{\mathbf{w}^{rel,T} L [I - \tilde{A} - r K]^{-1} \mathbf{x}} \quad \in (0, 1)
	\end{equation}
	Using this scalar, total profit simplifies greatly to $\Pi(r) = F(\mathbf{c}) \cdot \gamma(r)$.
	
	\begin{theorem}[An algebraic condition for the second equality]\label{thm:two_equalities}
		Given that the first equality holds, ``total profit = total surplus value'' holds if and only if the chosen reduction vector $\mathbf{c}$ lies in the intersection of the value-feasible domain $\Theta^{val}$ and a particular affine hyperplane $H_2(r)$:
		\begin{equation}\label{eq:second_hyperplane}
			\mathbf{c}^T \boldsymbol{\eta}_{new}(r) = 0
		\end{equation}
		where the hyperplane normal vector is
		$\boldsymbol{\eta}_{new}(r) := L(I-\tilde{A})^{-1} \Big[ (\tilde{A}+BL)\mathbf{x} - (1-\gamma(r))\mathbf{x} \Big]$.
		\textit{(See the Appendix for the algebraic derivation.)}
	\end{theorem}
	
	The economic meaning of $\boldsymbol{\eta}_{new}(r)$ is profound. The first term inside the brackets, $(\tilde{A}+BL)\mathbf{x}$, represents \textbf{total physical reproduction requirements}, including depreciation, intermediate inputs, and workers' subsistence goods. The second term, $(1-\gamma(r))\mathbf{x}$, represents the \textbf{value remainder} left for sustaining social reproduction after deducting the capitalist profit share from total output. The joint validity of the two equalities requires that total physical requirements and the value remainder coincide exactly as scalars in the pure value space.
	
	\subsection{A compatibility interval}
	
	Does equation~\eqref{eq:second_hyperplane} truly intersect $\Theta^{val}$? Since $\Theta^{val}$ is a bounded closed set confined strictly to the first orthant, an arbitrary linear condition does not in general guarantee that the hyperplane will not ``miss'' the feasible set.
	
	To obtain a rigorous guarantee, we introduce a macro matrix operator capturing the extraction of pure surplus,
	$\hat{S} := L(I-\tilde{A}-BL)^{-1}$, and provide a critical condition for intersection.
	
	\begin{definition}[Critical profit shares based on the surplus operator]\label{def:critical_profit_rate}
		Let the composite extended material matrix be $\hat{A} = \tilde{A} + BL$. For a strictly positive output vector $\mathbf{x} > \mathbf{0}$, define the sector-$i$ critical profit share determining whether the transformation intersection occurs as:
		\begin{equation}\label{eq:r_crit_true}
			\gamma_i^{crit}(\mathbf{x}) := 1 - \frac{[\hat{S}\hat{A}\mathbf{x}]_i}{[\hat{S}\mathbf{x}]_i}
		\end{equation}
		Define the global lower and upper bounds as
		$\gamma_{\min}^{crit} := \min_{i} \gamma_i^{crit}$ and
		$\gamma_{\max}^{crit} := \max_{i} \gamma_i^{crit}$, respectively.
	\end{definition}
	
	\begin{theorem}[Universal compatibility of the two macro equalities]\label{thm:universal_compatibility}
		Assume the system features a physical surplus ($\lambda_{\max}(\tilde{A}+BL) < 1$). If the macro profit share $\gamma(r)$ of the actual economy lies in the compatibility interval $[\gamma_{\min}^{crit}, \gamma_{\max}^{crit}]$, then the hyperplane $H_2(r)$ must pass through the interior of the value-feasible domain $\Theta^{val}$. In this case, the solution set simultaneously satisfying the two macro equalities and the reproduction-floor constraints is a nonempty convex polytope of dimension $n-2$ (hence, infinitely many solutions exist).
	\end{theorem}
	\textit{(This theorem relies on the algebraic identity $(I-M_0)^{-1}L(I-\tilde{A})^{-1} = L(I-\tilde{A}-BL)^{-1}$; see the Appendix for the proof.)}
	
	Theorem~\ref{thm:universal_compatibility} establishes mathematically that, whenever the macro profit markup under capitalism falls within an objectively bounded interval, there necessarily exist reduction structures---within the broad space of social bargaining---that render the labor-value system fully consistent with the price-of-production system with stock-based profits.
	
	\subsection{Dialogue with the classical transformation literature}
	
	The two-layer mapping and compatibility-interval theory proposed in this paper provides an algebraically coherent account of the two macro equalities. Yet given a century of controversy over transformation in political economy, the economic content of this mathematical framework warrants further clarification. This subsection addresses several central theoretical questions.
	
	\textbf{First: the intrinsic link between distributive boundaries and ``labor creates value.''}
	At first glance, the value-feasible domain $\Theta^{val}$ may appear to be merely a calculation of ``how wages can vary across industries without breaking the system.'' This can create the impression that the model departs from the core proposition that labor creates value. In fact, the mathematical prerequisite for $\Theta^{val}$ to have a strictly positive interior is that the spectral radius of the labor-reproduction operator be strictly less than one ($\lambda_{\max}(M_0) < 1$). This is precisely a mathematical translation of Marx's law of surplus value in Volume I of \emph{Capital}: the distinctive use-value of labor power as a special commodity is that living labor (matrix $L$) creates, in production, a quantity of physical wealth that is objectively greater than the physical wealth required to reproduce itself (matrix $B$). Hence the breadth of the feasible domain is itself created by the value-creating capacity of living labor. The model shows that, although nominal wages are shaped by complex social struggles, such struggles are far from unconstrained; they are strictly bounded by the objective law that living labor must generate a physical surplus.
	
	\textbf{Second: dimensional conversion between time and money, and how this differs from the ``New Interpretation'' (NI).}
	In the main text we introduce the scalar $\kappa$ (MELT) to connect labor time (values) with legal tender (prices). Formally this resembles the NI proposed by Foley (1982) and Duménil (1983), yet the theoretical emphasis is fundamentally different. To make the macro equalities hold, NI no longer treats the value of labor power as an objective physical bundle; instead, it redefines it ex post as the money wage converted into ``representative labor'' via MELT. While this resolves the transformation inconsistency in distributional accounting, it relaxes the classical physical floor constraint that ``workers must consume specific material means of subsistence.''
	
	By contrast, this paper preserves the classical physical subsistence floor (the constant matrix $B$) and shifts theoretical flexibility onto the variable ``skilled-labor reduction coefficients'' determined by social processes. In our framework, $\kappa$ is used only after a structurally offsetting reduction vector has been found within the feasible domain allowed by objective boundaries; it then plays a purely dimensional role, filtering out the scale disturbance of monetary units and inflation.
	
	To clarify that ``$\kappa$ is not an extra condition introduced to force equality,'' the Appendix proves rigorously that even if the external multiplier $\kappa$ is entirely abandoned and the simple-labor normalization ($c_1=1$) is relaxed, the transformation problem remains solvable within our framework. In that case, the two macro equalities become two independent hyperplanes in a high-dimensional feasible space, and the compatibility interval still guarantees that these hyperplanes intersect within the interior of the physical reproduction-feasible region. Therefore, resolving the classical transformation problem requires neither sacrificing the objective physical floor nor relying on an external dimensional multiplier.
	
	It should be noted, however, that while dropping $\kappa$ is mathematically feasible, it can blur the economic interpretation: without $\kappa$ the degrees of freedom of the reduction vector are fully released, and no sector can be designated as simple labor, which makes the analysis more complex.
	
	\textbf{Third: cross-space mapping is not a tautology.}
	Does setting up a complex algebraic structure that makes the two equalities hold amount to a tautology? If the model could force equality under any parameters, it would indeed be mere mathematical play. However, our framework consists of two logically independent layers. The value-feasible domain $\Theta^{val}$ is generated purely from physical requirements $(\tilde{A}, B, L)$ and involves neither profit rates nor money. By contrast, the transformation hyperplane $H_2(r)$ is generated in the sphere of circulation from capitalists' pursuit of a stock-based profit rate $r$.
	
	Our derivations show that intersection between the two is not unconditional. When capitalists claim an excessively high profit rate ($r > r^*$), the hyperplane lies entirely outside the feasible domain, indicating that the system cannot transform values into prices while preserving labor-power reproduction. The existence of the compatibility interval $[\gamma_{\min}^{crit}, \gamma_{\max}^{crit}]$ demonstrates that only when the macro profit markup lies within objective bounds can the complex ``monetary profit distribution'' in circulation necessarily find a corresponding ``structure of surplus-labor extraction'' in the value space. The mapping is therefore conditional and reflects a deep structural property of the system, not a tautology.
	
	Samuelson's ``circular reasoning'' critique presupposes that reduction coefficients must be imported from the price space. Our price-invariance lemma and the purely physical construction of the feasible domain show that the existence space of the value system is independent of price formation, thus cutting off the alleged circularity at the entry point. The value category cannot be ``erased'' here because it provides information absent from price theory: the objective elastic boundary of distributive structures and the mechanism by which capital accumulation compresses that boundary.
	
	\textbf{Fourth: the economic meaning of the strictly positive lower bound of the compatibility interval ($\gamma_{\min}^{crit} > 0$).}
	If capitalists claim no profit at all ($r=0$), so that all money income becomes workers' nominal wages, why can macro consistency fail? Why is the lower bound of the compatibility interval strictly positive?
	
	If profit is zero in the price space ($\Pi=0$), then to satisfy the second equality, total surplus value $S$ in the value space must also be zero. Yet in our model, whenever social output exceeds the historically given subsistence floor implied by $B$, there exists an objective physical surplus, which implies that living labor time must exceed the floor value and hence $S$ is strictly positive. The only way to make $S=0$ would be for workers to transform their very high nominal wages entirely into additional material consumption, forcing the physical consumption matrix $B$ to expand until it exhausts the entire physical surplus. Therefore, as long as we treat the subsistence floor $B$ as relatively stable historically and the economy produces a surplus beyond that floor, the nominal price system must feature a strictly positive minimum macro profit share ($\gamma_{\min}^{crit} > 0$) to ``absorb'' and price that extra physical surplus.
	
	More fundamentally, this mathematical result touches a central contradiction of Marxian political economy: the unity of opposites between use-value (material content) and value (social form). At the physical level, the objective network produces ``surplus use-values'' beyond the subsistence floor of labor power; at the nominal level, if capitalists extract no profit ($r=0$), these additional material riches lose their corresponding ``value form'' in social relations. The strictly positive lower bound of the compatibility interval is precisely a mapping of this contradiction: as long as modern production objectively generates a physical surplus, its social value form must be secured through the extraction of a certain proportion of macro profit. This internal tension between the physical floor of production and nominal monetary distribution is not only a logical requirement for macro balancing in the algebraic system, but also a mathematical expression of the dialectics inherent in the labour theory of value.
	\section{Numerical Experiment}
	
	To illustrate the computational mechanism of the cross-space mapping theory, this section constructs a three-sector macroeconomic example. The example fully incorporates the fixed-capital stock matrix $K$, the depreciation matrix $D$, and the heterogeneous consumption structure $B$ (see the Appendix for the specific matrices). In three dimensions ($n=3$), once the benchmark normalization $c_1=1$ is imposed, the feasible domain is naturally reduced to a two-dimensional slice, providing a clear geometric perspective on the abstract algebraic structure.
	
	Given the specified physical network, the dominant eigenvalue of the system's labor-reproduction operator is $\lambda_{\max}(M_0) = 0.7184 < 1$, indicating that the value-feasible domain $\Theta^{val}$ exists and has a nonempty interior. The composite physical matrix incorporating workers' subsistence consumption implies that the maximum feasible profit rate is $r^* = 12.48\%$, while the purely technical maximum profit rate is $r_A = 44.77\%$. Their gap, $32.29\%$, measures the rigid constraint imposed by the reproduction floor on capital's valorization drive. Using the matrix identity developed in this paper, the critical profit-share interval ensuring both equalities is computed as $[10.69\%, 15.12\%]$.
	
	\begin{figure}[htbp]
		\centering
		\includegraphics[width=0.65\textwidth]{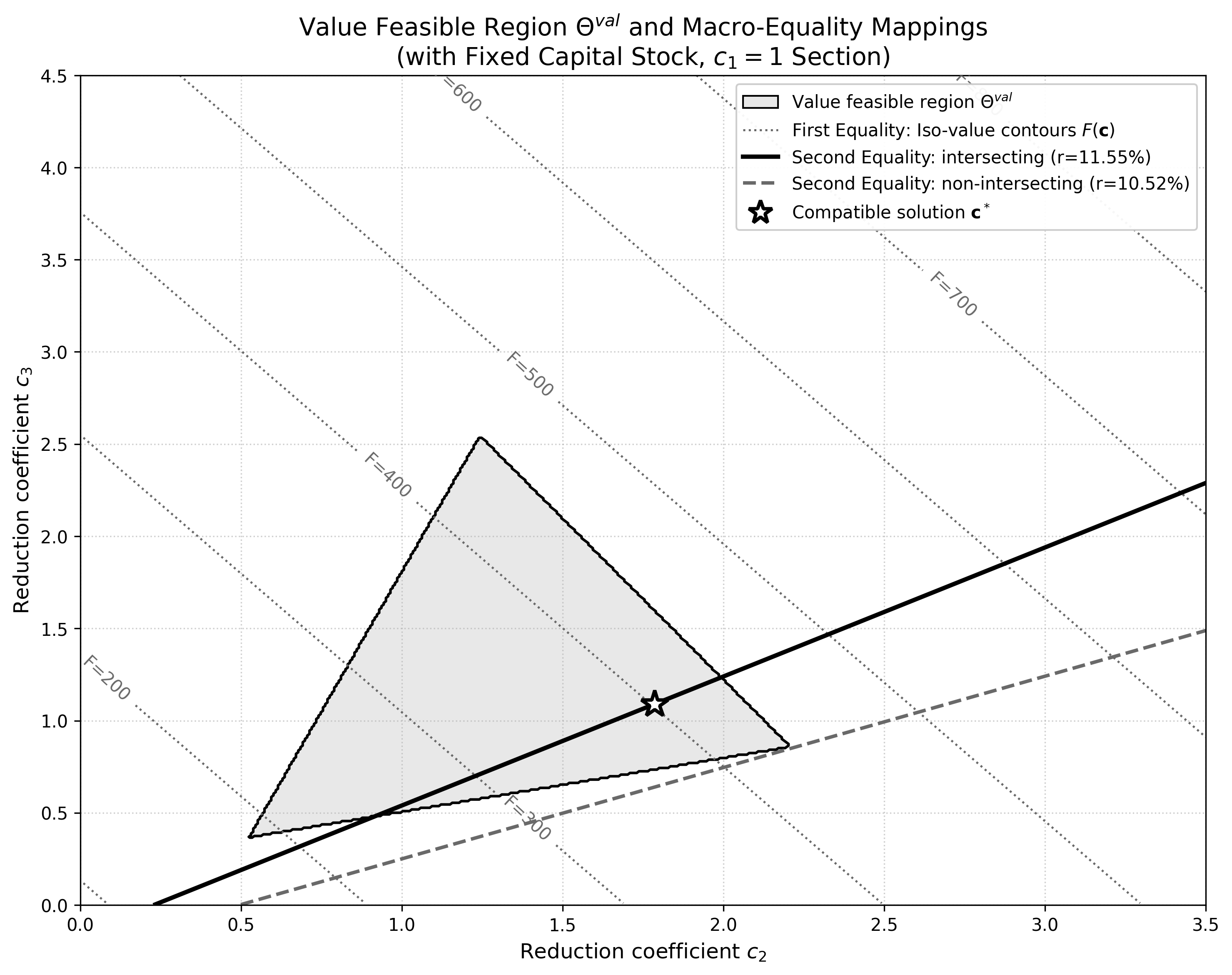}
		\caption{The value-feasible domain $\Theta^{val}$ and the geometric mapping of the two macro equalities (two-dimensional slice with $c_1=1$)}
		\label{fig:geometry}
	\end{figure}
	
	Figure~\ref{fig:geometry} displays, in a two-dimensional slice, the concrete geometry and mapping mechanism of the ``two macro equalities'' in the algebraic space.
	
	First, consider the physical subsistence floor and the value-feasible domain. The three linear inequalities expanded from the objective constraint $\mathbf{c}^T(I-M_0) \geq \mathbf{0}^T$ carve out a closed convex polygon (shaded in light gray). Any interior point $(c_2, c_3)$ represents a reduction vector that guarantees physical reproduction of labor power across all sectors.
	
	Second, consider the geometric meaning of the first equality. As shown by the fine gray dotted lines spanning the entire space, the first equality corresponds geometrically to a family of parallel iso-total-value lines. Each dotted line represents a particular scalar of aggregate total value. This makes clear that the first equality is not an exclusive boundary constraint but rather a background scaling that covers the whole space: for any point in the space, the system can compute an objectively existing total labor value and complete dimensional alignment through the unique multiplier $\kappa$.
	
	Finally, consider the constraint implied by the second equality. Given that the first equality has aligned dimensions, the second equality $\mathbf{c}^T \boldsymbol{\eta}_{new}(r) = 0$ appears as a straight line in the figure. As the profit rate $r$ varies, the intercept of this line shifts. When a test profit rate within the compatibility interval is chosen ($r=11.55\%$, corresponding to a profit share $\gamma = 11.71\%$), the thick black solid line representing the second equality passes exactly through the interior of the light-gray feasible domain. The segment of this solid line lying inside the polygon is precisely the solution set that ``both makes total profit equal total surplus value and preserves the workers' physical floor.'' By contrast, when the profit rate falls to $r=10.52\%$ (corresponding to $\gamma = 10.67\%$, below the lower compatibility bound), the gray dashed line representing the second equality grazes the feasible domain from outside, yielding no admissible intersection.
	
	Selecting an intersection point on the segment, $\mathbf{c}^* = (1.0000, 1.7868, 1.0902)^T$, and substituting it into the model, the numerical results show that the relative errors of ``total price vs.\ total value'' and ``total profit vs.\ total surplus value'' both converge to the computational precision limit of $10^{-9}$. This visualization confirms that treating the skilled-labor reduction coefficients as a multidimensional vector with a meaningful interior is key to achieving macro-level consistency between the law of value and the nominal price system.
	
	\section{Conclusion and Discussion}
	
	Skilled-labor reduction and the transformation problem are classical analytical difficulties in political economy. Much of the existing literature focuses on finding a unique constant solution for reduction proportions within a single algebraic system. This paper proposes a change in perspective: a feasible-domain framework that shifts the focus from ``solving for a single number'' to ``characterizing the distributive space under objective constraints imposed by the physical production network.''
	
	The theoretical analysis yields three main results. First, as long as the macroeconomy can generate a physical surplus, the set of reduction coefficients compatible with the social reproduction floor forms a bounded closed convex set (the feasible domain). Second, capital accumulation and a rising profit rate generate an objective crowding-out effect that strictly contracts the economy-wide feasible range of nominal wage choices. Third, by introducing the macro profit share and an algebraic mapping across the two-layer spaces, the macro equalities ``total price equals total value'' and ``total profit equals total surplus value'' can intersect and become logically consistent within a reasonable compatibility interval.
	
	Our framework offers several implications for understanding the labour theory of value and related macro distributive questions:
	\begin{enumerate}[label=(\roman*)]
		\item \textbf{The law of value appears as an inclusive objective boundary.} The underlying physical production network does not rigidly dictate a unique conversion proportion; rather, it delineates an objective space of action for class struggle and institutional arrangements. As long as the actual distributive system remains within this boundary, macro-level equivalent exchange and surplus-value theory can still hold.
		\item \textbf{Stock capital plays a non-negligible role in macro value analysis.} In empirical work, incorporating the advance of constant-capital stocks into price accounting can help mitigate numerical biases in sectoral surplus measures that arise when relying solely on flow-based models.
		\item \textbf{The accumulation rate imposes a hard constraint on distributive flexibility.} The theory implies that, under given physical technology, an excessively high profit rate and accumulation rate will objectively limit the expansion of residents' consumption space. This provides an objective reference for coordinating real-economy development and preventing excessive distributive inequality.
	\end{enumerate}
	
	As an exploratory study, the present analysis has limitations. First, our theoretical characterization is built on a static input--output technology and does not yet examine the dynamic evolution of feasible domains driven by technical change (e.g., a rising organic composition of capital). Second, we focus primarily on macro-level compatibility; the micro-level mechanisms behind sectoral price fluctuations remain to be explored in future research.
	
	\bibliographystyle{unsrtnat}
	\bibliography{ref_domian}
	
	\newpage
	\appendix

	\renewcommand{\theequation}{\Alph{section}.\arabic{equation}}
	\renewcommand{\thetable}{\Alph{section}.\arabic{table}}
	\renewcommand{\thefigure}{\Alph{section}.\arabic{figure}}
	\setcounter{equation}{0}
	\setcounter{table}{0}
	\setcounter{figure}{0}
	
	\section{Mathematical Proofs}
	This appendix provides rigorous proofs of the main propositions, lemmas, and theorems stated in the main text.
	
	\subsection{Proof of Proposition~\ref{prop:dimensionless}: Dimensionlessness of $M_0$}
	
	\begin{proof}
		The $(i,j)$ element of $M_0 = L(I-\tilde{A})^{-1}B$ can be expanded by matrix multiplication as:
		\[
		[M_0]_{ij} = \sum_{k=1}^{n} L_{ii} [(I-\tilde{A})^{-1}]_{ik} B_{kj}
		= l_i \sum_{k=1}^{n} [(I-\tilde{A})^{-1}]_{ik} \beta_{kj}
		\]
		(where we use that $L$ is diagonal, hence all off-diagonal entries are zero). Consider the units of each factor:
		\begin{itemize}
			\item $l_i$: the direct labor time in sector $i$ required to produce one unit of commodity $i$, with units $[\text{hours} / \text{unit}_i]$;
			\item $[(I-\tilde{A})^{-1}]_{ik}$: the $(i,k)$ element of the composite Leontief inverse, representing the total (direct and indirect) gross output of commodity $i$ induced by one unit of final demand for commodity $k$, with units $[\text{unit}_i / \text{unit}_k]$;
			\item $\beta_{kj}$: the amount of commodity $k$ consumed per hour of work by workers in sector $j$, with units $[\text{unit}_k / \text{hour}]$.
		\end{itemize}
		
		Therefore, the dimensional product of each term inside the summation is:
		\[
		\frac{\text{hours}}{\text{unit}_i}
		\times \frac{\text{unit}_i}{\text{unit}_k}
		\times \frac{\text{unit}_k}{\text{hours}}
		= 1 \quad (\text{cancels out completely; dimensionless})
		\]
		Since every term in the series is dimensionless, their sum $[M_0]_{ij}$ is also a pure number. The proposition follows.
	\end{proof}

	\subsection{Proof of Lemma~\ref{lem:price_invariance}: Price invariance of $M_0$}
	
	\begin{proof}
		Let $\mathbf{p}>0$ be any strictly positive price vector, and define the diagonal price matrix $P=\mathrm{diag}(p_1,\ldots,p_n)$.
		By definition, converting physical matrices into monetary (value) matrices measured in the chosen price unit can be represented as similarity transformations/scalings:
		the composite material input matrix $\tilde{A}_{val}=P\tilde{A}P^{-1}$, the direct labor matrix $\tilde{L}=LP^{-1}$, and the consumption matrix $\tilde{B}=PB$.
		
		We first compute the inverse matrix $(I-\tilde{A}_{val})^{-1}$:
		\[
		I-\tilde{A}_{val} = I - P\tilde{A}P^{-1} = PIP^{-1} - P\tilde{A}P^{-1} = P(I-\tilde{A})P^{-1}.
		\]
		By standard properties of matrix inversion, it follows that $(I-\tilde{A}_{val})^{-1} = P(I-\tilde{A})^{-1}P^{-1}$.
		
		Next, substitute the transformed matrices into the definition of the labor-reproduction matrix:
		\begin{align*}
			M_{val} &= \tilde{L}(I-\tilde{A}_{val})^{-1}\tilde{B} \\
			&= (LP^{-1}) \cdot \left[ P(I-\tilde{A})^{-1}P^{-1} \right] \cdot (PB) \\
			&= L (P^{-1}P) (I-\tilde{A})^{-1} (P^{-1}P) B \\
			&= L(I-\tilde{A})^{-1}B \\
			&= M_0.
		\end{align*}
		Hence, for any price vector used as a scaling transformation, the diagonal matrices $P$ and $P^{-1}$ cancel out completely. Therefore, the extracted labor-reproduction matrix $M_0$ is exactly identical to the one computed in pure physical units. The lemma is proved.
	\end{proof}

	\subsection{Proof of Lemma~\ref{lem:M_spectrum}: Basic spectral properties of $M_0$}
	
	\begin{proof}
		\textbf{(i) Proof that $M_0$ is strictly positive ($M_0 > 0$):}
		
		\textit{Step 1: Show that $S := L(I-\tilde{A})^{-1} > 0$.}
		
		By Assumption~\ref{ass:A_L}, the composite material input matrix $\tilde{A}$ satisfies productivity, i.e., $\lambda_{\max}(\tilde{A}) < 1$.
		Hence $(I-\tilde{A})^{-1}$ exists and admits the Neumann series expansion:
		\[
		(I-\tilde{A})^{-1} = \sum_{k=0}^{\infty} \tilde{A}^k = I + \tilde{A} + \tilde{A}^2 + \cdots.
		\]
		For any element $(i, j)$, we show $\left[(I-\tilde{A})^{-1}\right]_{ij} > 0$.
		
		When $i = j$: the $k = 0$ term contributes $I_{ii} = 1 > 0$; since $\tilde{A} \geq 0$, all subsequent terms are nonnegative, hence $\left[(I-\tilde{A})^{-1}\right]_{ii} \geq 1 > 0$.
		
		When $i \neq j$: by Assumption~\ref{ass:A_L}, $\tilde{A}$ is irreducible. This implies that there exists at least one direct or indirect path from node $i$ to node $j$ in the production network. Therefore, there exists a path length $k^* \in \{1,\ldots,n-1\}$ such that $(\tilde{A}^{k^*})_{ij} > 0$. Hence,
		\[
		\left[(I-\tilde{A})^{-1}\right]_{ij} \geq (\tilde{A}^{k^*})_{ij} > 0.
		\]
		Thus $(I-\tilde{A})^{-1}$ is strictly positive (every entry is $>0$).
		Moreover, Assumption~\ref{ass:A_L} implies $L = \mathrm{diag}(l_1,\ldots,l_n)$ with all $l_j > 0$. Multiplying a strictly positive diagonal matrix by a strictly positive matrix yields a strictly positive matrix. Hence $S = L(I-\tilde{A})^{-1} > 0$.
		
		\textit{Step 2: Show that $M_0 = SB > 0$.}
		
		For any element $(i,j)$ of $M_0$,
		\[
		[M_0]_{ij} = (SB)_{ij} = \sum_{k=1}^n S_{ik} B_{kj}.
		\]
		From Step 1, $S_{ik} > 0$ for all $i,k$.
		By Assumption~\ref{ass:B}, the matrix $B$ has no all-zero column. Thus, for any fixed sector $j$, there exists at least one $k_0 \in \{1,\ldots,n\}$ such that $B_{k_0 j} > 0$.
		Since all terms are nonnegative, focusing on the $k_0$ term yields
		\[
		[M_0]_{ij} = \sum_{k=1}^n S_{ik} B_{kj} \geq S_{ik_0} B_{k_0 j} > 0.
		\]
		Therefore $M_{0,ij} > 0$ for all $i,j$, i.e., $M_0$ is strictly positive.
		
		\textbf{(ii) Irreducibility and aperiodicity:}
		
		Since $M_0$ is strictly positive (no zero entries), any two nodes in its associated directed graph are directly connected in both directions. Hence $M_0$ is irreducible. Moreover, the diagonal elements satisfy $[M_0]_{jj} > 0$, implying that each node has a self-loop. The existence of self-loops rules out periodicity; thus $M_0$ is aperiodic.
		
		\textbf{(iii) Spectral properties:}
		
		Given (i) and (ii), we may directly invoke the Perron--Frobenius theorem for strictly positive matrices: the dominant eigenvalue $\lambda^* = \lambda_{\max}(M_0)$ is a strictly positive real number and is simple. The corresponding left and right eigenvectors $\mathbf{y}^*$ and $\mathbf{x}^*$ have strictly positive components and are unique up to scalar multiplication. Furthermore, $\lambda^*$ strictly exceeds the modulus of any other (possibly complex) eigenvalue, i.e., a strict spectral gap exists. The lemma is proved.
	\end{proof}
	
	\subsection{Proof of Theorem~\ref{thm:existence}: Spectral equivalence for the existence of reduction coefficients}
	
	\begin{proof}
		\textbf{Part 1: Nonemptiness of the strictly feasible region (i) $\Leftrightarrow$ spectral-radius condition (ii)}
		
		\textit{Sufficiency:} Assume $\lambda^* = \lambda_{\max}(M_0) < 1$.
		Since $M_0$ is strictly positive, there exists a strictly positive left eigenvector $\mathbf{y}^* > \mathbf{0}$ such that
		$\mathbf{y}^{*T}M_0 = \lambda^*\mathbf{y}^{*T}$.
		Normalize it so that $y_1^*=1$, and use it as a constructed reduction vector, i.e., let $\mathbf{c}^* = \mathbf{y}^*$.
		Check the exploitation (subsistence-floor) condition:
		\[
		\mathbf{c}^{*T}(I-M_0) = \mathbf{y}^{*T} - \lambda^*\mathbf{y}^{*T} = (1-\lambda^*)\mathbf{y}^{*T} > \mathbf{0}^T.
		\]
		Because $\lambda^* < 1$ and $\mathbf{y}^* > \mathbf{0}$, the result is strictly positive. Hence $\mathbf{c}^*$ yields strictly positive exploitation in every sector, so $\mathbf{c}^* \in \operatorname{int}(\Theta^{val})$. Therefore the interior of the value-feasible domain is nonempty.
		
		\textit{Necessity:} Assume $\lambda^* \geq 1$.
		If $\lambda^* = 1$: for any candidate reduction vector $\mathbf{c} > \mathbf{0}$, by nonnegative matrix theory there exists a strictly positive right eigenvector $\mathbf{x}^* > \mathbf{0}$ such that $(I-M_0)\mathbf{x}^* = \mathbf{0}$.
		Consider the inner product:
		\[
		\mathbf{c}^T(I-M_0)\mathbf{x}^* = \mathbf{c}^T(\mathbf{0}) = 0.
		\]
		If the strictly feasible region were nonempty, there would exist some $\mathbf{c}$ such that $\mathbf{c}^T(I-M_0) > \mathbf{0}^T$. The inner product of a strictly positive row vector with a strictly positive column vector must be strictly positive, contradicting the equality above.
		
		Similarly, if $\lambda^* > 1$, then for the corresponding right eigenvector we have
		$(I-M_0)\mathbf{x}^* = (1-\lambda^*)\mathbf{x}^* < \mathbf{0}$.
		If there existed $\mathbf{c}$ satisfying the subsistence-floor condition $\mathbf{c}^T(I-M_0) \geq \mathbf{0}^T$, then its inner product with $\mathbf{x}^*$ would have to be nonnegative, but $\mathbf{c}^T(1-\lambda^*)\mathbf{x}^* < 0$, a contradiction. Hence the value-feasible domain must be empty.
		
		\textbf{Part 2: Labor-reproduction matrix (ii) $\Leftrightarrow$ extended physical matrix (iii)}
		
		We use a corollary of the Collatz--Wielandt theorem. The condition $\lambda_{\max}(\tilde{A}+BL) < 1$ is equivalent to the existence of some vector $\mathbf{u} > \mathbf{0}$ such that $(\tilde{A}+BL)\mathbf{u} < \mathbf{u}$.
		
		Expanding and rearranging yields $(I-\tilde{A})\mathbf{u} > BL\mathbf{u}$.
		Since $\tilde{A}$ is productive, $(I-\tilde{A})^{-1} \geq 0$. Left-multiplying both sides by $(I-\tilde{A})^{-1}$ preserves the inequality:
		\[
		\mathbf{u} > (I-\tilde{A})^{-1}BL\mathbf{u}.
		\]
		Introduce $\mathbf{h} = L\mathbf{u}$. Because $L > 0$ (positive diagonal) and $\mathbf{u} > \mathbf{0}$, we have $\mathbf{h} > \mathbf{0}$.
		Substitute $\mathbf{u} = L^{-1}\mathbf{h}$ into the inequality:
		\[
		L^{-1}\mathbf{h} > (I-\tilde{A})^{-1}B\mathbf{h}.
		\]
		Left-multiply both sides by the positive diagonal matrix $L$:
		\[
		\mathbf{h} > L(I-\tilde{A})^{-1}B\mathbf{h} = M_0\mathbf{h}.
		\]
		The above algebraic steps are reversible. Therefore,
		\[
		\exists\,\mathbf{u}>\mathbf{0}: (\tilde{A}+BL)\mathbf{u}<\mathbf{u}
		\quad\Longleftrightarrow\quad
		\exists\,\mathbf{h}>\mathbf{0}: M_0\mathbf{h}<\mathbf{h}.
		\]
		Consequently, the dominant eigenvalues of the two matrices are simultaneously less than 1. The theorem is proved.
	\end{proof}
	
	\subsection{Proof of Proposition~\ref{prop:convex}: Bounded closed convexity of $\Theta^{val}$}
	
	After fixing the normalization $c_1=1$, we must show that the remaining variables
	$\boldsymbol{\theta} = (c_2, \ldots, c_n)^T$ form a bounded, closed, and convex set.
	We first introduce an auxiliary lemma about diagonal entries.
	
	\textbf{Lemma A.1} (Diagonal-entry property):
	\textit{Let $M_0 > 0$ and $\lambda^* = \lambda_{\max}(M_0) \leq 1$. Then the diagonal entries satisfy $[M_0]_{jj} < 1$ for all $j = 1,\ldots,n$.}
	
	\textit{Proof of Lemma A.1:}
	Let $\mathbf{y}^T M_0 = \lambda^* \mathbf{y}^T$ with $\mathbf{y} > \mathbf{0}$. Expanding the $j$-th component yields:
	\[
	\sum_{i=1}^n y_i [M_0]_{ij} = \lambda^* y_j.
	\]
	Separating the diagonal term gives
	$y_j [M_0]_{jj} + \sum_{i \neq j} y_i [M_0]_{ij} = \lambda^* y_j$.
	Since $M_0 > 0$ and $\mathbf{y} > \mathbf{0}$, the sum over $i\neq j$ is strictly positive. Therefore,
	\[
	[M_0]_{jj} = \lambda^* - \frac{1}{y_j}\sum_{i \neq j} y_i [M_0]_{ij} < \lambda^* \leq 1.
	\]
	This completes the proof.
	
	\begin{proof}[Proof of Proposition~\ref{prop:convex}]
		\textbf{Convexity and closedness:}
		Let $\boldsymbol{\theta}_1, \boldsymbol{\theta}_2 \in \bar{\Theta}^{val}$ and $\alpha \in [0,1]$.
		By linearity of matrix multiplication,
		\[
		(1, \alpha\boldsymbol{\theta}_1 + (1-\alpha)\boldsymbol{\theta}_2)(I-M_0)
		= \alpha(1,\boldsymbol{\theta}_1)(I-M_0) + (1-\alpha)(1,\boldsymbol{\theta}_2)(I-M_0) \geq \mathbf{0}^T.
		\]
		Thus the solution set is convex. Moreover, the constraints use weak inequalities ($\geq$), defining an intersection of half-spaces; hence the boundary is included and the set is closed.
		
		\textbf{Boundedness (existence of upper and lower bounds):}
		Expand the constraint $\mathbf{c}^T(I-M_0) \geq \mathbf{0}^T$ column by column.
		
		For the constraint corresponding to column 1:
		\[
		(1-[M_0]_{11}) - \sum_{i=2}^n [M_0]_{i1}c_i \geq 0
		\quad \Rightarrow \quad
		\sum_{i=2}^n [M_0]_{i1}c_i \leq 1 - [M_0]_{11}.
		\]
		Since each term in the sum is positive, keeping only the $j$-th term preserves the inequality:
		\[
		[M_0]_{j1}c_j \leq \sum_{i=2}^n [M_0]_{i1}c_i \leq 1-[M_0]_{11}.
		\]
		By Lemma A.1, $1-[M_0]_{11} > 0$. Hence,
		\[
		c_j \leq \frac{1-[M_0]_{11}}{[M_0]_{j1}} =: c_{\max,j} < \infty.
		\]
		This establishes a finite upper bound in each dimension.
		
		For the constraint corresponding to column $j$ ($j \geq 2$):
		\[
		(1-[M_0]_{jj})c_j - [M_0]_{1j} - \sum_{i=2, i\neq j}^n [M_0]_{ij}c_i \geq 0.
		\]
		Dropping the nonnegative summation term yields
		\[
		(1-[M_0]_{jj})c_j \geq [M_0]_{1j} > 0.
		\]
		Using Lemma A.1 again gives
		\[
		c_j \geq \frac{[M_0]_{1j}}{1-[M_0]_{jj}} =: c_{\min,j} > 0.
		\]
		This establishes a strictly positive lower bound in each dimension.
		
		In sum, $\Theta^{val}$ has strictly positive lower bounds and finite upper bounds in all dimensions and includes its boundary; therefore it is a bounded closed convex set. The proposition is proved.
	\end{proof}
	
	\subsection{Proof of Theorem~\ref{thm:equilibrium}: The equal-exploitation point}
	
	\begin{proof}
		\textbf{Interior location:}
		By the proof of the existence theorem, taking $\mathbf{c} = \mathbf{y}^*$ gives
		$\mathbf{y}^{*T}(I-M_0) = (1-\lambda^*)\mathbf{y}^{*T} > \mathbf{0}^T$.
		Since the inequality is strict, the point lies in the interior of the value-feasible domain.
		
		\textbf{Equal exploitation rates:}
		Substitute $\mathbf{y}^*$ into the exploitation-rate formula for sector $j$:
		\[
		e_j(\mathbf{y}^*) = \frac{y_j^* - [\mathbf{y}^{*T}M_0]_j}{[\mathbf{y}^{*T}M_0]_j}.
		\]
		Because $\mathbf{y}^*$ is the left eigenvector associated with $\lambda^*$, we have
		$[\mathbf{y}^{*T}M_0]_j = \lambda^* y_j^*$.
		Substitute to simplify:
		\[
		e_j(\mathbf{y}^*) = \frac{y_j^* - \lambda^* y_j^*}{\lambda^* y_j^*}
		= \frac{1-\lambda^*}{\lambda^*} = e^* > 0.
		\]
		Since $\lambda^*$ is constant, the result is identical for all sectors, implying a uniform exploitation rate $e^*$.
		
		\textbf{Uniqueness:}
		Suppose there exists a reduction vector $\mathbf{c} > \mathbf{0}$ (with $c_1=1$) such that all sectors share the same exploitation rate $e$. Then necessarily,
		\[
		c_j = (1+e)[\mathbf{c}^T M_0]_j \quad \forall j.
		\]
		In vector form, $\mathbf{c}^T = (1+e)\mathbf{c}^T M_0$, i.e.,
		$\mathbf{c}^T M_0 = \frac{1}{1+e}\mathbf{c}^T$.
		Thus $\mathbf{c}^T$ must be a positive left eigenvector of $M_0$ with eigenvalue $\frac{1}{1+e}$.
		By matrix theory, an irreducible positive matrix has a unique positive eigenvector (up to scale), and it corresponds to the dominant eigenvalue $\lambda^*$.
		Therefore $\frac{1}{1+e} = \lambda^*$ (equivalently $e = e^*$), and the normalized reduction vector must be $\mathbf{y}^*$. The theorem is proved.
	\end{proof}
	
	\subsection{Proof of Proposition~\ref{prop:collapse}: Degeneration of the feasible domain}
	
	\begin{proof}
		In the zero-surplus case, $\lambda_{\max}(M_0) = 1$.
		Let $\mathbf{c} \in \Theta^{val}$ be any reduction vector satisfying the subsistence floor, i.e., $\mathbf{c}^T(I-M_0) \geq \mathbf{0}^T$.
		Since $\lambda_{\max}(M_0) = 1$, there exists a strictly positive right eigenvector $\mathbf{x}^* > \mathbf{0}$ such that $(I-M_0)\mathbf{x}^* = \mathbf{0}$.
		
		Consider the inner product:
		\[
		\mathbf{c}^T(I-M_0)\mathbf{x}^* = \mathbf{c}^T (\mathbf{0}) = 0.
		\]
		The inequality implies $\mathbf{v}^T := \mathbf{c}^T(I-M_0)$ is a nonnegative vector, while $\mathbf{x}^*$ is strictly positive. For the sum $\sum v_j x_j^* = 0$ to hold, every component of $\mathbf{v}^T$ must be zero. Hence,
		\[
		\mathbf{c}^T(I-M_0) = \mathbf{0}^T
		\quad \Rightarrow \quad
		\mathbf{c}^T M_0 = \mathbf{c}^T.
		\]
		Thus, in this boundary case, any $\mathbf{c}$ satisfying the subsistence constraint must be a left eigenvector of $M_0$ with eigenvalue 1.
		Since the positive eigenvector is unique up to normalization, the value-feasible domain collapses to a singleton $\mathbf{y}^*$, i.e., $\Theta^{val} = \{\mathbf{y}^*\}$.
		
		Substituting into the exploitation-rate expression:
		\[
		e_j = \frac{y_j^* - [\mathbf{y}^{*T}M_0]_j}{[\mathbf{y}^{*T}M_0]_j}
		= \frac{y_j^* - 1\cdot y_j^*}{1\cdot y_j^*} = 0.
		\]
		Hence all exploitation rates are zero: value created is fully used to reproduce labor power. The proposition is proved.
	\end{proof}
	
	\subsection{Proof of Proposition~\ref{prop:M_r_properties}: Monotonic properties of the parametric reproduction matrix}
	
	\begin{proof}
		\textbf{(i) Strict positivity and continuity:}
		By the definition of the maximum technical profit rate $r_A$, for any $r \in [0, r_A)$ the dominant eigenvalue of $\tilde{A} + rK$ is strictly less than 1. Then the inverse matrix in the price system admits the Neumann series:
		\[
		R(r) := [I - (\tilde{A} + rK)]^{-1} = \sum_{k=0}^{\infty} (\tilde{A} + rK)^k.
		\]
		Since $\tilde{A}$ is an irreducible nonnegative matrix and $K \geq 0$, the sum $(\tilde{A} + rK)$ is also irreducible and nonnegative. Hence its inverse $R(r)$ has strictly positive entries, i.e., $R(r)>0$. Because $L$ is a positive diagonal matrix and $B$ is nonnegative with no all-zero column, the product $M(r) = L R(r) B$ is strictly positive, i.e., $M(r) > 0$.
		Moreover, since the spectral-radius condition remains satisfied, matrix inversion is continuous in $r$, and thus all entries of $M(r)$ depend continuously on $r$.
		
		\textbf{(ii) Elementwise strict monotonicity:}
		Let $0 \leq r_1 < r_2 < r_A$. Then $(\tilde{A} + r_1K) \leq (\tilde{A} + r_2K)$.
		Since $K$ is nonnegative and not identically zero, the difference $(r_2 - r_1)K$ is a nonnegative nonzero matrix.
		Using the series expansion,
		\[
		R(r_2) - R(r_1) = \sum_{k=1}^{\infty} \left[ (\tilde{A} + r_2K)^k - (\tilde{A} + r_1K)^k \right].
		\]
		Given irreducibility, this difference is not only nonnegative but, through accumulation over the infinite series, becomes strictly positive elementwise. Hence $R(r_1) < R(r_2)$ (elementwise strict inequality).
		Since $L$ has strictly positive diagonal entries and $B$ has strictly positive relevant entries, pre- and post-multiplication preserves strict inequality, yielding $M(r_1) < M(r_2)$. Thus each element increases strictly with $r$.
		
		\textbf{(iii) Monotonicity and divergence of the dominant eigenvalue:}
		By a corollary of the Collatz--Wielandt theorem for nonnegative matrices, $M(r_1) < M(r_2)$ (with both strictly positive) implies
		$\lambda_{\max}(M(r_1)) < \lambda_{\max}(M(r_2))$; thus the dominant eigenvalue increases strictly with $r$.
		As $r \to r_A$, the dominant eigenvalue of $(\tilde{A} + rK)$ approaches 1, forcing the elements of $R(r)$ to diverge, and hence the elements of $M(r)$ diverge. Consequently, $\lambda_{\max}(M(r)) \to +\infty$. The proposition is proved.
	\end{proof}
	
	\subsection{Proof of Theorem~\ref{thm:r_star}: Existence of the maximum feasible profit rate}
	
	\begin{proof}
		Define $f: [0, r_A) \to \mathbb{R}$ by $f(r) := \lambda_{\max}(M(r))$.
		By Proposition~\ref{prop:M_r_properties}, $f(r)$ satisfies:
		(1) continuity on its domain;
		(2) strict monotonic increase;
		(3) boundary condition $f(0) = \lambda_{\max}(M(0)) < 1$ (by the assumption of a positive surplus);
		(4) divergence $\lim_{r \to r_A} f(r) = +\infty$.
		
		By the intermediate value theorem, there exists a unique $r^* \in (0, r_A)$ such that $f(r^*) = 1$, i.e., $\lambda_{\max}(M(r^*)) = 1$.
		Since $f(r)$ is strictly increasing, for any $r > r^*$ we have $\lambda_{\max}(M(r)) > 1$.
		By the existence theorem in the main text, the wage set compatible with physical reproduction becomes empty. Hence $r^*$ is the maximum profit rate a capitalist economy can sustain. The theorem is proved.
	\end{proof}
	
	\subsection{Proof of Theorem~\ref{thm:duality}: Duality between the profit rate and distributive sets}
	
	\begin{proof}
		\textbf{Part 1: Strict monotone contraction $\Theta^{price}(r_2) \subsetneq \Theta^{price}(r_1)$}
		
		Let $0 \leq r_1 < r_2 \leq r^*$. We first show $\Theta^{price}(r_2) \subseteq \Theta^{price}(r_1)$.
		Take any $\mathbf{w} \in \Theta^{price}(r_2)$, i.e., $\mathbf{w} > \mathbf{0}$ and $\mathbf{w}^T[I - M(r_2)] \geq \mathbf{0}^T$.
		By Proposition~\ref{prop:M_r_properties}, $M(r_1) < M(r_2)$ elementwise, hence $I - M(r_1) > I - M(r_2)$.
		Rewrite:
		\[
		\mathbf{w}^T[I - M(r_1)] = \mathbf{w}^T[I - M(r_2)] + \mathbf{w}^T[M(r_2) - M(r_1)].
		\]
		Since $\mathbf{w} > \mathbf{0}$ and $[M(r_2) - M(r_1)] > 0$, their product is a row vector with strictly positive components. Therefore,
		\[
		\mathbf{w}^T[I - M(r_1)] > \mathbf{w}^T[I - M(r_2)] \geq \mathbf{0}^T,
		\]
		implying $\mathbf{w} \in \Theta^{price}(r_1)$.
		
		Next, the inclusion is strict. Since $\Theta^{price}(r_1)$ has a nonempty interior, there exists a point $\mathbf{w}^*$ on its (relative) boundary such that, for some coordinate $j_0$, the constraint binds:
		$\bigl(\mathbf{w}^{*T}[I - M(r_1)]\bigr)_{j_0} = 0$.
		Consider the corresponding constraint under $r_2$:
		\[
		\bigl(\mathbf{w}^{*T}[I - M(r_2)]\bigr)_{j_0}
		= 0 - \sum_{i=1}^n w_i^* \bigl[M(r_2) - M(r_1)\bigr]_{ij_0} < 0.
		\]
		Thus $\mathbf{w}^* \notin \Theta^{price}(r_2)$, proving strict inclusion.
		
		\textbf{Part 2: Singleton degeneration}
		
		When $r = r^*$, $\lambda_{\max}(M(r^*)) = 1$.
		By symmetry with the zero-surplus boundary case analyzed in Proposition~\ref{prop:collapse}, any strictly positive vector satisfying
		$\mathbf{w}^T(I - M(r^*)) \geq \mathbf{0}^T$ must in fact satisfy equality and hence must be a left eigenvector of $M(r^*)$ associated with eigenvalue 1.
		Since the positive eigenvector is unique after normalization, $\Theta^{price}(r^*)$ degenerates to a singleton. The theorem is proved.
	\end{proof}
	
	\subsection{Proof of Theorem~\ref{thm:two_equalities}: An algebraic condition for the second equality}
	
	\begin{proof}
		The second equality requires total profit to equal total surplus value, i.e., $\Pi(r) = S(\mathbf{c})$.
		
		As stated in the main text, total profit can be written as a fraction of total price:
		\[
		\Pi(r) = r \mathbf{p}^T K \mathbf{x}
		= \kappa \cdot r \mathbf{w}^{rel,T} L [I - \tilde{A} - rK]^{-1} K \mathbf{x}.
		\]
		Substitute $\kappa = \frac{\mathbf{c}^T L(I-\tilde{A})^{-1}\mathbf{x}}{\mathbf{w}^{rel,T} L [I-\tilde{A}-rK]^{-1}\mathbf{x}}$ and use the definition of $\gamma(r)$:
		\[
		\Pi(r) = \left[ \mathbf{c}^T L(I-\tilde{A})^{-1}\mathbf{x} \right] \cdot \gamma(r).
		\]
		
		Meanwhile, total surplus value in the pure value space expands as:
		\begin{align*}
			S(\mathbf{c})
			&= \mathbf{c}^T L\mathbf{x} - \mathbf{c}^T M_0 L\mathbf{x} \\
			&= \mathbf{c}^T L(I-\tilde{A})^{-1}(I-\tilde{A})\mathbf{x}
			- \mathbf{c}^T L(I-\tilde{A})^{-1}BL\mathbf{x} \\
			&= \mathbf{c}^T L(I-\tilde{A})^{-1}[(I-\tilde{A}-BL)\mathbf{x}].
		\end{align*}
		
		Equate the two expressions:
		\[
		\mathbf{c}^T L(I-\tilde{A})^{-1}\mathbf{x} \cdot \gamma(r)
		= \mathbf{c}^T L(I-\tilde{A})^{-1} \left[ (I-\tilde{A}-BL)\mathbf{x} \right].
		\]
		Move the left-hand term to the right and factor out $\mathbf{c}^T L(I-\tilde{A})^{-1}$:
		\[
		\mathbf{c}^T L(I-\tilde{A})^{-1}
		\Big[ (I-\tilde{A}-BL)\mathbf{x} - \gamma(r)\mathbf{x} \Big] = 0.
		\]
		Rearrange the bracket:
		\[
		(I-\tilde{A}-BL)\mathbf{x} - \gamma(r)\mathbf{x}
		= \mathbf{x} - (\tilde{A}+BL)\mathbf{x} - \gamma(r)\mathbf{x}
		= (1-\gamma(r))\mathbf{x} - (\tilde{A}+BL)\mathbf{x}.
		\]
		For interpretability, extract a minus sign (which does not affect the equality to zero):
		\[
		\mathbf{c}^T L(I-\tilde{A})^{-1}
		\Big[ (\tilde{A}+BL)\mathbf{x} - (1-\gamma(r))\mathbf{x} \Big] = 0.
		\]
		Define $\boldsymbol{\eta}_{new}(r) := L(I-\tilde{A})^{-1}
		\Big[ (\tilde{A}+BL)\mathbf{x} - (1-\gamma(r))\mathbf{x} \Big]$.
		Then $\mathbf{c}^T \boldsymbol{\eta}_{new}(r) = 0$, completing the proof.
	\end{proof}
	
	\subsection{Proof of Theorem~\ref{thm:universal_compatibility}: Universal compatibility and a matrix identity}
	
	\begin{proof}
		To show that the hyperplane $\mathbf{c}^T\boldsymbol{\eta}_{new}(r)=0$ must intersect the interior of $\Theta^{val}$, we show that the intersection lies within the polyhedral cone induced by the constraint $\mathbf{c}^T(I-M_0) \ge \mathbf{0}^T$.
		
		\textbf{Step 1: An algebraic mapping of the value-feasible domain} \\
		Define a nonnegative vector $\mathbf{q} \geq \mathbf{0}$ by $\mathbf{q}^T = \mathbf{c}^T (I - M_0)$.
		Since the system has a surplus, $\lambda_{\max}(M_0) < 1$, so $(I-M_0)^{-1}$ exists and is strictly positive (or at least nonnegative). Hence any feasible reduction vector can be parameterized as:
		\[
		\mathbf{c}^T = \mathbf{q}^T (I - M_0)^{-1}, \quad \mathbf{q} \geq \mathbf{0}.
		\]
		Substitute into the hyperplane condition; transformation compatibility is equivalent to finding $\mathbf{q}\ge 0$ such that:
		\[
		\mathbf{q}^T (I - M_0)^{-1} \boldsymbol{\eta}_{new}(r) = 0.
		\]
		
		\textbf{Step 2: A composite matrix identity} \\
		Expanding $\boldsymbol{\eta}_{new}(r)$, the core operator becomes $(I - M_0)^{-1} L(I-\tilde{A})^{-1}$.
		Let $S = L(I-\tilde{A})^{-1}$, so that $M_0 = SB$. We claim and prove:
		\[
		(I - SB)^{-1} S = L (I - \tilde{A} - BL)^{-1} := \hat{S}.
		\]
		\textit{Proof:} Right-multiply the left-hand side by $(I - \tilde{A} - BL)$:
		\begin{align*}
			(I - SB)^{-1} S (I - \tilde{A} - BL)
			&= (I - SB)^{-1} [S(I-\tilde{A}) - SBL] \\
			&= (I - SB)^{-1} [L(I-\tilde{A})^{-1}(I-\tilde{A}) - SBL] \\
			&= (I - SB)^{-1} [L - SBL] \\
			&= (I - SB)^{-1} (I - SB) L = L.
		\end{align*}
		By uniqueness of the inverse, the identity holds.
		
		\textbf{Step 3: Pinning down the compatibility interval} \\
		Using the identity, the intersection problem in $\mathbf{c}$-space simplifies to an equation in $\mathbf{q}$-space:
		\[
		\mathbf{q}^T \left( \hat{S}\hat{A}\mathbf{x} - (1-\gamma(r))\hat{S}\mathbf{x} \right) = 0,
		\]
		where $\hat{S} = L(I-\tilde{A}-BL)^{-1} > 0$ and $\hat{A} = \tilde{A}+BL$.
		Let the bracketed vector be $\tilde{\boldsymbol{\eta}}(r)$.
		By the Collatz--Wielandt extremal characterization: for any positive vector
		$\mathbf{y} = \hat{S}\mathbf{x} > \mathbf{0}$, the dominant eigenvalue $\lambda_{\max}(\hat{A})$ lies between the minimum and maximum of the componentwise ratios
		$[\hat{S}\hat{A}\mathbf{x}]_i / [\hat{S}\mathbf{x}]_i$.
		
		By definition,
		$\gamma_i^{crit} = 1 - \frac{[\hat{S}\hat{A}\mathbf{x}]_i}{[\hat{S}\mathbf{x}]_i}$.
		If the realized profit share $\gamma(r)$ lies between $\gamma_{\min}^{crit}$ and $\gamma_{\max}^{crit}$, then the scalar $(1-\gamma(r))$ lies between the minimum and maximum of the ratios. This implies that $\tilde{\boldsymbol{\eta}}(r)$ must have \textbf{mixed signs} (both positive and negative components).
		
		\textbf{Step 4: Existence of an interior intersection} \\
		Since $\tilde{\boldsymbol{\eta}}(r)$ has mixed signs, there exists a vector $\mathbf{q}^* > \mathbf{0}$ such that $\mathbf{q}^{*T} \tilde{\boldsymbol{\eta}}(r) = 0$.
		Map back to reduction coefficients via $\mathbf{c}^* = (I - M_0)^{-T} \mathbf{q}^*$.
		Because $\mathbf{q}^* > \mathbf{0}$ and the inverse is nonnegative, $\mathbf{c}^*$ satisfies the strict subsistence-floor condition $\mathbf{c}^{*T}(I-M_0) > \mathbf{0}^T$.
		Therefore, whenever the profit share lies in the compatibility interval, the equality hyperplane passes through the interior of $\Theta^{val}$. The theorem is proved.
	\end{proof}
	
	\section{A Transformation Proof Without an Explicit Dimensional Multiplier}
	
	In the main-text solution, we fixed sector 1 as the unit benchmark ($c_1 = 1$) to pin down the relative scale of reduction coefficients, and introduced the scalar multiplier $\kappa$ (MELT) to align macro aggregates dimensionally. This appendix investigates a more demanding theoretical claim: if we abandon the external multiplier $\kappa$ and reject any人为 dimensional alignment, can the law of value still be compatible---within the algebraic system itself---with the two macro equalities in the nominal price-of-production system?

	Once $\kappa$ is dropped, values (labor time) must be made numerically equal to prices of production (legal tender) directly.
	
	In the price space, given a stock-based profit rate $r$ and sectoral nominal wages $\mathbf{w}$, capitalists generate an objectively existing vector of prices of production $\mathbf{p}^*$. Under this price system, two macro constants exist:
	1. Total social price: $P^* = \mathbf{p}^{*T} \mathbf{x}$; \\
	2. Total stock-based profit: $\Pi^* = r \mathbf{p}^{*T} K \mathbf{x}$. \\
	The realized macro profit share is then the constant $\gamma^* = \Pi^* / P^*$.
	
	In the pure value space, to ``absorb'' legal tender directly, we must relax the restriction $c_1 = 1$. The reduction vector $\mathbf{c} = (c_1, c_2, \ldots, c_n)^T > \mathbf{0}$ now regains $n$ degrees of freedom and must satisfy the objective physical subsistence floor:
	$\mathbf{c}^T (I - M_0) \geq \mathbf{0}^T$.

	Without $\kappa$, for the two equalities to hold exactly, the vector $\mathbf{c}$ must satisfy simultaneously:
	\begin{align}
		\text{First equality:} \quad & \mathbf{c}^T L(I-\tilde{A})^{-1} \mathbf{x} = P^* \label{eq:app_D_eq1} \\
		\text{Second equality:} \quad & \mathbf{c}^T L(I-\tilde{A})^{-1}(I-\tilde{A}-BL)\mathbf{x} = \Pi^* \label{eq:app_D_eq2}
	\end{align}
	Thus, in an $n$-dimensional variable space, the two macro equalities define two non-parallel affine hyperplanes. Solving the transformation problem is equivalent to proving that the intersection line of these two hyperplanes must pass through the interior of the value-feasible domain $\Theta^{val}$.
	
	To prove this, define $\mathbf{q}^T = \mathbf{c}^T(I-M_0)$. The reproduction floor implies $\mathbf{q} \geq \mathbf{0}$.
	Since the system has a physical surplus ($\lambda_{\max}(M_0) < 1$), the invertible transformation
	$\mathbf{c}^T = \mathbf{q}^T (I-M_0)^{-1}$ holds.
	
	Substitute this into the two equalities, and use the matrix identity proved above,
	$(I-M_0)^{-1} L(I-\tilde{A})^{-1} = L(I-\tilde{A}-BL)^{-1} = \hat{S}$, to obtain:
	
	\textbf{Substitute into the first equality \eqref{eq:app_D_eq1}:}
	\[
	\mathbf{q}^T [ (I-M_0)^{-1} L(I-\tilde{A})^{-1} ] \mathbf{x} = P^*
	\quad \Longrightarrow \quad
	\mathbf{q}^T (\hat{S} \mathbf{x}) = P^*.
	\]
	
	\textbf{Substitute into the second equality \eqref{eq:app_D_eq2}:}
	\[
	\mathbf{q}^T [ (I-M_0)^{-1} L(I-\tilde{A})^{-1} ] (I-\tilde{A}-BL)\mathbf{x} = \Pi^*.
	\]
	Replace the bracketed term by $\hat{S}$:
	\[
	\mathbf{q}^T \hat{S} (I-\tilde{A}-BL)\mathbf{x} = \Pi^*.
	\]
	By the definition of $\hat{S}$, it is the inverse operator of $(I-\tilde{A}-BL)$ in the relevant composition, yielding:
	\[
	\mathbf{q}^T (L \mathbf{x}) = \Pi^*.
	\]

	After the transformation, the two macro equalities become a purely linear system in the nonnegative mapping space ($\mathbf{q} \geq \mathbf{0}$). In matrix form:
	\begin{equation}
		\begin{bmatrix} \mathbf{v}_1^T \\ \mathbf{v}_2^T \end{bmatrix} \mathbf{q}
		= \begin{bmatrix} P^* \\ \Pi^* \end{bmatrix},
	\end{equation}
	where $\mathbf{v}_1 = \hat{S} \mathbf{x} > \mathbf{0}$ and $\mathbf{v}_2 = L \mathbf{x} > \mathbf{0}$ are known strictly positive vectors.
	
	Whether this nonnegative linear system admits a strictly positive solution $\mathbf{q} > \mathbf{0}$ is essentially a generation problem for a \emph{polyhedral convex cone}. By \textbf{Farkas' lemma} (and its convex-cone equivalent form), the system $V^T \mathbf{q} = \mathbf{b}$ admits a strictly positive solution $\mathbf{q} > \mathbf{0}$ if and only if the target vector $\mathbf{b}$ lies strictly in the interior of the convex cone generated by the columns of $V^T$ (i.e., the vectors $(\mathbf{v}_{1i}, \mathbf{v}_{2i})^T$).
	
	In the two-dimensional plane (spanned by $P^*$ and $\Pi^*$), this is equivalent to requiring that the ratio $\Pi^* / P^*$ lies strictly between the minimum and maximum slopes of the generating rays, i.e., between
	$\min_i \frac{[\mathbf{v}_2]_i}{[\mathbf{v}_1]_i}$ and $\max_i \frac{[\mathbf{v}_2]_i}{[\mathbf{v}_1]_i}$.
	Provided $n \geq 3$ and there is no degenerate collinearity, the interior of this cone is a broad wedge region, and the two hyperplanes necessarily have a large $(n-2)$-dimensional intersection within the positive orthant (hence infinitely many solutions $\mathbf{q}^* > \mathbf{0}$).
	
	Consider this ratio interval. Expanding components:
	\[
	\text{ratio}_i = \frac{[L\mathbf{x}]_i}{[\hat{S}\mathbf{x}]_i}.
	\]
	Using the operator expansion $\hat{S}\mathbf{x} = L\mathbf{x} + L(\tilde{A}+BL)\hat{S}\mathbf{x}$, substitute and simplify:
	\[
	\text{ratio}_i
	= 1 - \frac{[L(\tilde{A}+BL)\hat{S}\mathbf{x}]_i}{[\hat{S}\mathbf{x}]_i}
	= 1 - \frac{[\hat{S}\hat{A}\mathbf{x}]_i}{[\hat{S}\mathbf{x}]_i}.
	\]
	This coincides with the ``critical profit shares'' $\gamma_i^{crit}$ defined in Definition~\ref{def:critical_profit_rate}.
	
	Therefore, if the realized profit share $\gamma^* = \Pi^* / P^*$ lies in the compatibility interval
	$[\gamma_{\min}^{crit}, \gamma_{\max}^{crit}]$, then there exist infinitely many $\mathbf{q}^*$ in the strictly nonnegative mapping space such that both equalities hold. Mapping back via
	$\mathbf{c}^* = (I-M_0)^{-T} \mathbf{q}^*$ yields a reduction vector that (i) numerically absorbs monetary units without any external multiplier, (ii) satisfies Marx's two macro equalities, and (iii) strictly guarantees the physical reproduction floor for the working class.
	
	Having established that the two hyperplanes must intersect even without an external multiplier, we now prove that the absolute reduction vector $\mathbf{c}^*$ obtained in the ``no-multiplier'' absolute system is proportional to the relative vector $\mathbf{c}_{rel}$ used in the main text (normalized by $c_{rel,1}=1$).
	
	\begin{proof}
		Fix a compatible profit rate $r$. Suppose there exists a relative reduction vector $\mathbf{c}_{rel}$ satisfying the transformation intersection condition in the relative system (with $c_{rel,1} = 1$).
		By the compatibility condition in the main text, $\mathbf{c}_{rel}$ must satisfy:
		\begin{equation}\label{eq:proof_ratio}
			\frac{S(\mathbf{c}_{rel})}{F(\mathbf{c}_{rel})}
			= \frac{\mathbf{c}_{rel}^T \hat{S} (I-\tilde{A}-BL)\mathbf{x}}{\mathbf{c}_{rel}^T \hat{S} \mathbf{x}}
			= \gamma^* = \frac{\Pi^*}{P^*},
		\end{equation}
		where $P^*$ is the empirically observed total price, $\Pi^*$ is total stock-based profit, and $\hat{S}$ is the surplus operator.
		
		Define a macro scalar $\mu$, interpreted as the ratio between ``actual total price'' and ``relative total labor value'' (i.e., MELT $\kappa$):
		\begin{equation}
			\mu = \frac{P^*}{F(\mathbf{c}_{rel})} = \frac{P^*}{\mathbf{c}_{rel}^T \hat{S} \mathbf{x}}.
		\end{equation}
		
		Now construct a test vector in the absolute system:
		$\hat{\mathbf{c}} = \mu \cdot \mathbf{c}_{rel}$.
		We verify that it satisfies both macro equalities.
		
		\textbf{Check the first equality:}
		\begin{align*}
			\hat{\mathbf{c}}^T \hat{S} \mathbf{x}
			&= (\mu \cdot \mathbf{c}_{rel})^T \hat{S} \mathbf{x}
			= \mu \left( \mathbf{c}_{rel}^T \hat{S} \mathbf{x} \right).
		\end{align*}
		Substitute the definition of $\mu$:
		\[
		= \left( \frac{P^*}{\mathbf{c}_{rel}^T \hat{S} \mathbf{x}} \right)
		\left( \mathbf{c}_{rel}^T \hat{S} \mathbf{x} \right) = P^*.
		\]
		Thus the first equality holds exactly.
		
		\textbf{Check the second equality:}
		\begin{align*}
			\hat{\mathbf{c}}^T \hat{S} (I-\tilde{A}-BL)\mathbf{x}
			&= \mu \left( \mathbf{c}_{rel}^T \hat{S} (I-\tilde{A}-BL)\mathbf{x} \right).
		\end{align*}
		By \eqref{eq:proof_ratio},
		$\mathbf{c}_{rel}^T \hat{S} (I-\tilde{A}-BL)\mathbf{x}
		= \frac{\Pi^*}{P^*} (\mathbf{c}_{rel}^T \hat{S} \mathbf{x})$.
		Hence,
		\[
		= \mu \left[ \frac{\Pi^*}{P^*} \left( \mathbf{c}_{rel}^T \hat{S} \mathbf{x} \right) \right].
		\]
		Substitute $\mu$ again:
		\[
		= \left( \frac{P^*}{\mathbf{c}_{rel}^T \hat{S} \mathbf{x}} \right)
		\left[ \frac{\Pi^*}{P^*} \left( \mathbf{c}_{rel}^T \hat{S} \mathbf{x} \right) \right]
		= \Pi^*.
		\]
		Thus the second equality holds exactly.
		
		\textbf{Check the physical reproduction floor (feasibility constraint):}
		Since $\mathbf{c}_{rel} \in \Theta^{val}$, we have $\mathbf{c}_{rel}^T (I-M_0) \geq \mathbf{0}^T$.
		Also $\mu>0$ because both total price and total value are positive. By positive scaling of linear inequalities:
		\[
		(\mu \mathbf{c}_{rel})^T (I-M_0) \geq \mathbf{0}^T
		\quad \Longrightarrow \quad
		\hat{\mathbf{c}}^T (I-M_0) \geq \mathbf{0}^T.
		\]
		Hence $\hat{\mathbf{c}}$ remains feasible.
		
		\textbf{Conclusion:}
		The constructed vector $\hat{\mathbf{c}}$ satisfies all transformation conditions in the absolute system without any external variables. Since the intersection of the two hyperplanes defines the solution affine subspace, the solution vector $\mathbf{c}^*$ must coincide with $\hat{\mathbf{c}}$. Therefore,
		\begin{equation}
			\mathbf{c}^* = \mu \cdot \mathbf{c}_{rel}.
		\end{equation}
		Thus, even in the no-multiplier absolute system, the solution $\mathbf{c}^*$ preserves the relative distributive structure encoded by $\mathbf{c}_{rel}$; the extra scalar $\mu$ endogenously absorbs the dimensional conversion from labor time into legal tender.
	\end{proof}
	\section{Parameterization and Algebraic Details for the Numerical Experiment}
	The example describes an economy with fixed capital. The circulating-capital input matrix $A$ (current intermediate-input use) and the fixed-capital stock matrix $K$ (advanced equipment stocks) are:
	\[
	A = \begin{bmatrix}
		0.15 & 0.18 & 0.12 \\
		0.20 & 0.12 & 0.15 \\
		0.10 & 0.15 & 0.18
	\end{bmatrix}, \quad
	K = \begin{bmatrix}
		0.40 & 0.35 & 0.30 \\
		0.30 & 0.45 & 0.25 \\
		0.25 & 0.30 & 0.40
	\end{bmatrix}.
	\]
	The depreciation-rate vector is $\boldsymbol{\delta} = (0.10, 0.12, 0.08)^T$, generating the depreciation matrix
	$D = K \cdot \mathrm{diag}(\boldsymbol{\delta})$.
	The composite material input matrix is:
	\[
	\tilde{A} = A + D = \begin{bmatrix}
		0.190 & 0.222 & 0.144 \\
		0.230 & 0.174 & 0.170 \\
		0.125 & 0.186 & 0.212
	\end{bmatrix}.
	\]
	Let $L = \mathrm{diag}(0.40, 0.60, 0.35)$ be the direct-labor diagonal matrix. The heterogeneous consumption matrix $B$ (column $j$ is the physical consumption basket of workers in sector $j$) and the gross output vector $\mathbf{x}$ are:
	\[
	B = \begin{bmatrix}
		0.25 & 0.30 & 0.20 \\
		0.20 & 0.25 & 0.15 \\
		0.22 & 0.28 & 0.25
	\end{bmatrix}, \quad
	\mathbf{x} = \begin{bmatrix}
		100 \\ 80 \\ 120
	\end{bmatrix}.
	\]
	
	The computed dominant eigenvalues are:
	\begin{itemize}
		\item Composite material input matrix: $\lambda_{\max}(\tilde{A}) = 0.5519 < 1$ (Hawkins--Simon satisfied);
		\item Basic labor-reproduction matrix: $\lambda_{\max}(M_0) = 0.7184 < 1$ (nonempty value-feasible domain);
		\item Extended physical matrix: $\lambda_{\max}(\tilde{A} + BL) = 0.8749 < 1$ (physical surplus exists).
	\end{itemize}
	The implied profit-rate bounds are: maximum technical profit rate $r_A = 44.77\%$ and maximum feasible profit rate $r^* = 12.48\%$.
	The critical profit-share interval ensuring transformation compatibility is
	$[\gamma_{\min}^{crit}, \gamma_{\max}^{crit}] = [10.69\%, 15.12\%]$.
	
	At the profit rate $r = 11.55\%$ (corresponding to a profit share $\gamma = 11.71\%$), the intersection point obtained via constrained optimization is
	$\mathbf{c}^* = (1.0000, 1.7868, 1.0902)^T$. Substituting this solution into the model yields the following verification results:
	\begin{itemize}
		\item Total labor value: $F(\mathbf{c}^*) = 401.5912$
		\item Total surplus value: $S(\mathbf{c}^*) = 47.0379$
		\item Dimensional conversion multiplier: $\kappa = 0.9386$
		\item Total price: $P(\mathbf{w}) = 401.5912$
		\item Total profit: $\Pi = r \cdot \mathbf{p}^T K \mathbf{x} = 47.0379$
		\item Error in the first equality: $|P - F| = 0$
		\item Error in the second equality: $|\Pi - S| = 3.54 \times 10^{-9}$
	\end{itemize}
	The sectoral exploitation rates are $e_1 = 9.23\%$, $e_2 = 57.42\%$, and $e_3 = 37.04\%$, exhibiting pronounced cross-sector heterogeneity. This reflects structural differences across labor types between value creation and the reproduction costs of labor power.
\end{document}